\documentclass{article}

  \usepackage[preprint]{neurips_2026}

\usepackage[utf8]{inputenc} 
\usepackage[T1]{fontenc}    
\usepackage{hyperref}
\usepackage{url}            
\usepackage{booktabs}       
\usepackage{amsfonts}       
\usepackage{nicefrac}       
\usepackage{microtype}      
\usepackage{xcolor}         
\usepackage{algorithm}
\usepackage{algorithmic}

\usepackage{titletoc} 

\usepackage{xcolor}
\definecolor{myGreen}{rgb}{0, .6, .0}
\definecolor{myorange}{RGB}{251, 122, 56}
\definecolor{goodblue}{HTML}{0071bc}
\hypersetup{colorlinks=true,breaklinks,linkcolor=red,citecolor=goodblue}
\definecolor{softblue}{rgb}{0.85, 0.95, 1}

\title{Online Price Competition under Generalized Linear Demands}

\author{%
Daniele Bracale\thanks{Corresponding author.} \\
  Department of Statistics\\
    University of Michigan\\
  Ann Arbor, MI \\
  \texttt{dbracale@umich.edu} \\
   \And
   Moulinath Banerjee \\
   Department of Statistics \\
   University of Michigan \\
   \texttt{moulib@umich.edu} \\
   \AND
   Cong Shi \\
   Department of Management \\
   University of Miami \\
   \texttt{congshi@bus.miami.edu} \\
   \And
   Yuekai Sun \\
   Department of Statistics \\
   University of Michigan \\
   \texttt{yuekai@umich.edu} \\
}

\usepackage{nicefrac}
\mathchardef\mhyphen="2D
\newcommand{\nML}{\tt PML}
\newcommand{\ML}{\tt PML \mhyphen GLUCB}

\usepackage{enumitem}

\usepackage{amsmath}
\usepackage{amssymb}
\usepackage{mathtools}
\usepackage{amsthm}

\usepackage[capitalize,noabbrev]{cleveref}

\theoremstyle{plain}
\newtheorem{theorem}{Theorem}[section]
\theoremstyle{plain}
\newtheorem{proposition}[theorem]{Proposition}
\theoremstyle{plain}
\newtheorem{lemma}[theorem]{Lemma}
\theoremstyle{plain}

\theoremstyle{plain}
\theoremstyle{definition}

\theoremstyle{plain}
\newtheorem{assumption}[theorem]{Assumption}
\theoremstyle{remark}
\theoremstyle{plain}
\newtheorem{remark}[theorem]{Remark}
\theoremstyle{plain}
\newtheorem{example}[theorem]{Example}

\begin{document}

\maketitle

\begin{abstract}
We study a sequential price competition among $N$ sellers, each influenced by the pricing decisions of their rivals. Specifically, the demand function for each seller $i$ follows the single index model $\lambda_i(\mathbf p) = \mu_i(\langle \boldsymbol \theta_{i,0}, \mathbf p \rangle)$, with known increasing link $\mu_i$ and unknown parameter $\boldsymbol \theta_{i,0}$, where the vector $\mathbf{p}$ denotes the vector of prices offered by all the sellers simultaneously at a given instant. Each seller observes only their own realized demand — unobservable to competitors — and the prices set by rivals. We propose a novel decentralized policy, PML-GLUCB, that combines penalized MLE with an upper-confidence pricing rule. Our approach (i) \emph{removes the need for coordinated front-loaded exploration phases across sellers} — which is integral to previous models — making our method aligned with realistic market conditions; (ii) generalizes existing approaches that focus solely on linear demand models; (iii) accommodates both binary and real-valued demand observations. Relative to a dynamic benchmark policy, each seller achieves $\widetilde{O}(\sqrt{T})$ regret, which matches the optimal rate known in the linear setting. 
\end{abstract}

\section{Introduction}

A pricing strategy aims to determine prices that maximize current profits. Much of the learning-to-price literature, however, has traditionally focused on monopolistic settings \citep{fan2024policy, javanmard2019dynamic, golrezaei2019dynamic, bracale2025dynamic}, treating demand as independent of competitors’ actions. In many real-world markets, this assumption is unrealistic. Markets typically feature multiple sellers offering products that differ in quality and price, each seeking to maximize its own revenue. In such environments, a seller’s profit is inherently affected by the pricing decisions of its competitors. As a result, firms must continuously revise their prices to anticipate and respond to rivals’ moves. Motivated by these considerations, recent work has moved towards the study of pricing policies in multi-seller environments, explicitly modeling strategic interactions among competing firms \citep{li2024lego, bracale2025revenue, yang2024competitive, goyal2023learning, zhu2025dynamic}. In such settings, each seller seeks to maximize its own revenue, which inherently depends on the pricing policies adopted by its competitors. This interdependence naturally leads to the concept of a Nash equilibrium (NE): a profile of prices at which no seller can improve its payoff through a unilateral deviation, given the strategies of the others. Building on this framework, existing works \citep{li2024lego, bracale2025revenue} propose sequential pricing mechanisms designed to learn equilibrium behavior over time, and establish sublinear regret guarantees relative to the equilibrium benchmark\footnote{In this paper, we use the terms \emph{sublinear regret} and \emph{no regret} interchangeably to indicate that the cumulative regret grows sub-linearly with time, that is, $\mathrm{Reg}(T) = o(T)$ as $T \to \infty$.} and converge to a NE \citep{li2024lego, bracale2025revenue}. Yet two practical issues remain:
\begin{enumerate}[topsep=1pt, itemsep=1pt, parsep=0pt, leftmargin=1.5em]
    \item Existing approaches in the competitive setting \citep{li2024lego, bracale2025revenue} typically decouple exploration from optimization and rely on an \emph{explicit front-loaded exploration phase with an i.i.d. pricing policy}. As noted by \citet{jin2021double}, classical results such as \citet{garivier2016explore} show that the \emph{``explore-then-commit strategy is suboptimal in the asymptotic sense as the horizon grows, and thus, is worse than fully sequential strategies such as Upper Confidence Bound''} (UCB). Moreover, in some dynamic-pricing implementations, the required length of the exploration phase \emph{depends on unknown problem-dependent constants}, which complicates deployment in practice (see e.g., \cite{bracale2025revenue} Remark 5.5).
    \item The demand is often forced to be linear (or narrowly parameterized) \citep{bertsimas2006dynamic, guda2026demand, li2024lego, goyal2023learning}, limiting the ability to capture nonlinear demand patterns commonly observed in real markets \citep{gallego2006price,wan2022nonlinear,wan2023conditions}.
\end{enumerate}

To address these limitations, we adopt a generalized linear (single-index) demand framework, which accommodates a broader range of \emph{nonlinear demand} responses and naturally covers \emph{both binary and continuous demand outcomes}, and propose a UCB pricing policy that operates \emph{without a dedicated initial i.i.d. pricing-design exploration phase: the UCB strategy automatically induces the exploration.}

\subsection{Related Literature and Main Idea}\label{sec:Related-Literature}

\textbf{Demand learning in competitive environments.} Over the past two decades, there has been substantial progress in demand learning under competition. Early work by \citet{kirman1975learning} analyzed symmetric two-seller competition with linear demand. Asymmetry and noise were incorporated by \citet{cooper2015learning}, who let each seller learn independently (effectively as a monopolist) and established equilibrium under known price sensitivities. In contrast, \citet{li2024lego} studied dynamic pricing with \emph{unknown} price sensitivities under a linear demand model and achieved the optimal $\widetilde{O}(\sqrt{T})$ regret rate. In their setup, seller $i$’s expected demand is  $\lambda_i(\mathbf{p})=\alpha_{i,0}-\beta_{i,0}p_i+\sum_{j \in [N] \setminus \{i\}}\gamma_{ij,0}p_j$, where $[N] = \{1, 2, \dots, N\}$ indexes the sellers and the true parameters $(\alpha_{i,0}, \beta_{i,0}, \gamma_{ij,0})$ are unknown with $\beta_{i,0}>0$. The condition $\beta_{i,0} > 0$ captures the standard economic assumption that the demand of seller $i$ decreases as their own price $p_i$ increases. However, nonlinear demand models are often more realistic and have been explored in several applications \citep{gallego2006price, wan2022nonlinear, wan2023conditions}. For instance, \citet{gallego2006price} considered the mean demand function $\lambda_i(\mathbf{p})=a_i\left(p_i\right)/(\sum_j a_j\left(p_j\right)+\kappa)$, for some \textit{known} increasing functions $a_i$ and real value $\kappa \in (0,1]$. They establish the existence of a NE, but they do not propose an online learning procedure.

In contrast to prior studies \citep{guda2026demand, li2024lego, kachani2007modeling, gallego2006price}, we analyze a more general \emph{monotone single-index} expected-demand model:
\begin{equation}\label{eq:SIM}
\lambda_i(\mathbf{p}) =\mu_i(-\beta_{i,0}p_i+\langle \boldsymbol{\gamma}_{i,0},\mathbf{p}_{-i}\rangle)=\mu_i(\langle \boldsymbol{\theta}_{i,0},\mathbf{p}\rangle),
\end{equation}
for each $i \in [N]$, where $\mathbf{p}_{-i} = (p_j)_{j \in [N] \setminus\{i\}}$, $\mathbf{p}=  (p_j)_{j \in [N]}$ and $\boldsymbol{\theta}_{i,0}$ is a vector of dimension $N$ with $i$-th entry equals to $-\beta_{i,0}$ (where $\beta_{i,0}>0$) and the rest of the entries are the values of $\boldsymbol{\gamma}_{i,0}$ ordered. The link function $\mu_i$ is assumed to be increasing, capturing the same economic intuition that demand decreases with the seller’s own price $p_i$. Here $\mu_i$ is known to seller $i$ (but unknown to competitors) and the parameter $\boldsymbol{\theta}_{i,0}$ is unknown (to everyone). This model generalizes the linear demand model by \citet{li2024lego}; specifically, $\mu_i(u)=u+\alpha_{i,0}$ recovers their formulation.


\textbf{Multi-Agent Reinforcement Learning (MARL).} MARL is a broad framework for sequential decision-making problems involving multiple strategic agents that at each instant select actions based on their local observations, receive feedback in the form of rewards (which depend on personal and competitors' actions), and then update their policy to maximize long-term returns. The ultimate objective is to reach a stable solution concept, such as an NE, where no agent can unilaterally improve its expected return. A key modeling assumption in MARL algorithms is full or partial observability of other agents' rewards or payoffs \citep{hu2003nash, zhang2021multi, yang2020overview}. For instance, the canonical formulation of \citet{hu2003nash} assumes that each agent can observe the immediate reward of all players and use this to update joint policy. Similarly, \citet{zhang2021multi} and \citet{yang2020overview} rely on either joint reward observability or access to sufficient statistics (e.g., payoff signals or utility vectors) that allow agents to infer the equilibrium structure. This assumption is critical: observing other agents' realized rewards enables an agent to estimate the full payoff landscape and to learn best-response dynamics with respect to the joint reward function. Our setting fundamentally differs from this MARL paradigm in both information structure and observability. In the sequential pricing game we study, each seller observes only their own reward (or demand value) and rivals’ posted prices -- never competitors’ demands or revenues. Because demand functions are private and unobservable, a seller cannot recover rivals’ payoff functions or utility gradients, so classical MARL techniques are not directly applicable.

\textbf{\underline{Main idea}: connection to adversarial contextual bandits and a timing asymmetry.}
Our key idea is to interpret $\mathbf{p}_{-i}^{(t)}$ as the \emph{context} faced by seller $i$ at round $t$, thereby connecting our setting to the literature on contextual bandits \citep{chu2011contextual, agrawal2013thompson, chakraborty2023thompson, kim2023double} and, more broadly, to nonlinear Cournot competition \citep{cournot1838recherches, saloner1987cournot}. Our model, however, differs from the standard contextual bandit framework in two important respects. First, in contextual bandits the context is typically exogenous, whereas here $\mathbf{p}_{-i}^{(t)}$ is endogenous: it evolves as a function of past prices and therefore of past demand shocks. Second, in classical contextual bandits the context is observed before the action is chosen, while in our market sellers move simultaneously, so seller $i$ observes $\mathbf{p}_{-i}^{(t)}$ only after choosing $p_i^{(t)}$. This timing asymmetry rules out a direct application of standard contextual-bandit analyses, particularly those based on the elliptical potential lemma \citep{carpentier2020elliptical}, which underlies many regret bounds. To overcome this difficulty, we adapt the analysis to the simultaneous-move setting and derive a suitable elliptical-type bound (see \cref{lemma:Elliptical-Lemma}) that still delivers regret control under delayed context observation.

\subsection{Contributions}\label{sec:contributions}
We summarize our main contributions.

\textbf{(1) Front-loaded exploration-free competitive learning.}
We propose a fully decentralized policy, 
$\ML$, that \emph{learns while pricing} and eliminates the need for an explicit, front-loaded exploration phase, in contrast to prior approaches \citep{li2024lego, bracale2025revenue}. This design directly addresses key practical limitations of explore–then–commit strategies, which are often artificial (since, in real markets, \emph{``experimenting with prices is usually not feasible or desirable”} \citep{cohen2015pricing}), may depend on unknown problem-specific constants (\cite{bracale2025revenue} Remark 5.5), and can be asymptotically suboptimal relative to fully sequential learning schemes \citep{garivier2016explore}. Instead, 
our $\ML$ algorithm continuously balances estimation and revenue via an optimistic (UCB-based) best-response update.

\textbf{(2) Generalize linear demand with private feedback.}
We study a competitive demand-learning model of the form
\(
\lambda_i(\mathbf p)=\mu_i(\langle \boldsymbol{\theta}_{i,0},\mathbf p\rangle),
\)
where $\mu_i$ is a known increasing link and $\boldsymbol{\theta}_{i,0}$ is unknown.
This strictly generalizes the linear models \citep{li2024lego,guda2026demand} and, in contrast to prior works, naturally covers both (i) \emph{binary} outcomes (purchase / no-purchase) \citep{goyal2023learning,golrezaei2019dynamic,javanmard2017perishability,javanmard2019dynamic,fan2024policy,bracale2025dynamic} and (ii) \emph{continuous} demand observations \citep{bertsimas2006dynamic, li2024lego, bracale2025revenue} within the same analysis. Crucially, our information structure matches real markets: each seller observes only \emph{their own} realized demand/revenue and competitors’ posted prices, not competitors’ outcomes, ruling out standard multi-agent RL assumptions (as noted above in \cref{sec:Related-Literature}) and requiring new proof techniques.

\textbf{(3) Optimal regret and new concentration tools.}
We establish that, using our $\ML$ algorithm, each seller achieves $\operatorname{Reg}_i(T)=\widetilde{O}(N^{2}\sqrt{T})$ against a dynamic benchmark, matching the optimal $\widetilde{O}(\sqrt{T})$ dependence on $T$ known from linear competitive pricing, where $\widetilde{O}$ suppresses logarithmic factors.. On the technical side, we derive an \emph{elliptical-potential-type bound} tailored to simultaneous-move competition. As a byproduct, we also obtain a quantitative bound on aggregate deviation from equilibrium, $\nicefrac{1}{T}\sum_{t=1}^T \mathbb{E}\|\mathbf p^{(t)}-\mathbf p^{\star}\|_2^2=\widetilde{O}(\nicefrac{N^2}{\sqrt{T}})$, showing that regret-optimal play keeps the market close to the Nash price vector $\mathbf p^{\star}$ on average over time.

\textbf{Notation.} We write $[N]=\{1,\dots,N\}$ and $[T]=\{1,\dots,T\}$. We denote by $\mathbf{p} = (p_j)_{j \in [N]}$ the vector of prices set by all sellers, and by $\mathbf{p}_{-i} = (p_j)_{j \in [N] \setminus \{i\}}$ the vector of prices set by seller $i$’s competitors. When we write $\mathbf{p} = (p_i, \mathbf{p}_{-i})$, we do so to emphasize that the price $p_i$ is chosen by seller $i$; this notation does not imply that $p_i$ is the first component of $\mathbf{p}$ -- it remains in its $i$-th position. We use $\|\cdot\|_p$ for the $\ell^p$, Euclidean norm, and $\langle \mathbf{u},\mathbf{v}\rangle=\mathbf{u}^\top\mathbf{v}$ for the inner product. We write $\widetilde{O}(\cdot)$ to suppress logarithmic factors and use $\lesssim$ to hide absolute constants. For a vector $\mathbf{p} \in \mathbb{R}^N$ and a symmetric positive definite matrix $V \in \mathbb{R}^{N\times N}$ we define the norm $\lVert \mathbf{p} \rVert_V := \sqrt{\mathbf{p}^{\top} V \mathbf{p}}$. A function $f(\mathbf{p}): \mathbb{R}^N \mapsto \mathbb{R}$ is $L$-Lipschitz (with respect to a norm $\|\cdot\|$), for some $L>0$ if $|f\left(\mathbf{p}^{\prime}\right)- f(\mathbf{p})| \leq L \| \mathbf{p}^{\prime}- \mathbf{p}\|$ for all $\mathbf{p}, \mathbf{p}^{\prime}$ in the domain of $f$. We use $\nabla f (\mathbf{p})$ and $\nabla^2 f (\mathbf{p})$ to denote the gradient vector and the Hessian matrix of $f$ at $\mathbf{p}$. We denote by $\mathrm{C}^m(\Omega)$ the set of $m$-times continuously differentiable functions $f:\Omega \rightarrow \mathbb{R}$. A vector field $\mathbf{F}: \Omega \mapsto \Omega$ is a contraction w.r.t. some norm $\|\cdot\|$ if it is $L$-Lipschitz with $L<1$.

\textbf{Organization of the paper.} The remainder of the paper is organized as follows. Section~\ref{sec:prob_form} introduces the problem formulation and model assumptions, and formally defines regret and Nash equilibrium. Section~\ref{sec:algo} presents our proposed methods, including the penalized maximum likelihood estimator and the UCB-based pricing policy. Section~\ref{sec:concentration} derives concentration bounds and establishes our main regret guarantees, outlining the key steps of the proofs. Section~\ref{sec:conclusion} concludes the paper by discussing limitations and future research directions. Additional technical details and complete proofs are deferred to the supplementary material.

\section{Problem Setting}\label{sec:prob_form}

We consider a pricing problem faced by $N$ sellers that sell a single product with unlimited inventory sequentially over a time horizon $T$, where $T$ is \emph{known} to the sellers. Let $\mathbf{p}_{-i}^{(t)}\triangleq  (p_j^{(t)})_{j \in [N] \setminus\{i\}}$ denote the $i$-th seller's competitors' prices at time $t$ and $\mathbf{p}^{(t)}\triangleq (p_j^{(t)})_{j \in [N]}$ denote the joint prices. Only for the purpose of the analysis we define $\mathcal{P}\triangleq \prod_{i \in [N]}[\underline{p_i}, \bar{p}_i]$, where $0\le \underline{p_i}< \bar{p}_i<+\infty$, which denotes the support of joint prices and $\mathcal{P}_{-i}= \prod_{j \in [N]\setminus \{i\} }\mathcal{P}_j$ the support of the competitors' prices. Every seller knows their own price domain $\mathcal{P}_i$ but not the competing ones $\{\mathcal{P}_j\}_{j \in [N]\setminus \{i\} }$. However, only $B_p = (\sum_{i \in [N]}\bar{p}_i^2)^{1/2}$ is known to all the sellers. The general dynamic is as follows: at time $t=0$, each seller $i$ selects any initial price $p_i^{(0)} \in \mathcal{P}_i$ and everyone observes their own realized demand $y_i^{(0)}$ according to \eqref{eq:expon_model}. The price vector $\mathbf{p}^{(0)}$ is made public. Let $\mathcal{H}^{(0)}_i= \{(\mathbf{p}^{(0)},y_i^{(0)})\}$ be the data collected by seller $i$. For each sales round $t=1, \ldots, T$:

\begin{enumerate}[topsep=2pt, itemsep=2pt, parsep=0pt, leftmargin=1.5em]
  \item Each seller $i$ sets $p_i^{(t)}$ based on their history $\mathcal{H}_i^{(t-1)}$.
  \item Each seller $i$ observes their realized demand $y_i^{(t)}$, sampled by nature according to the model in \eqref{eq:expon_model}. The vector of prices $\mathbf{p}^{(t)}$ is made public, and then each seller updates their history: $\mathcal{H}_i^{(t)}=\mathcal{H}_i^{(t-1)}\cup \{(\mathbf{p}^{(t)},y_i^{(t)})\}$.
  \item Each seller $i$ observes their revenue $r_i = p_i^{(t)}y_i^{(t)}$.
\end{enumerate}

We emphasize that \emph{sellers can observe the historical prices of competitors, but do not know the demand of competitors}. We assume that $y_i^{(t)}$, given the price vector $\mathbf{p}^{(t)}$, is sampled by nature according to a \textit{canonical exponential family} w.r.t. a reference measure 
$\nu_i$:
\begin{equation}\label{eq:expon_model}
\textstyle \frac{d\mathbb{P}_{\boldsymbol{\theta}_{i,0}}(y_i | \mathbf{p})}{d\nu_i(y_i)} = \exp\left\{
y_i \langle \boldsymbol{\theta}_{i,0}, \mathbf{p}\rangle - b_i(\langle \boldsymbol{\theta}_{i,0}, \mathbf{p}\rangle) + c_i(y_i)\right\},
\end{equation}
where $c_i(\cdot)$ is a real-valued function and $b_i(\cdot)$ is assumed to be twice continuously differentiable. A random variable $y_i$ with the above density satisfies $\mathbb{E}(y_i| \mathbf{p}) = b'_i(\langle \boldsymbol{\theta}_{i,0}, \mathbf{p}\rangle)$ and $\textnormal{var}(y_i | \mathbf{p}) = b''_i(\langle \boldsymbol{\theta}_{i,0}, \mathbf{p}\rangle)$, showing that $b_i(\cdot)$ is strictly convex. The inverse link function $\mu_i:=b'_i$ is consequently strictly increasing. Thus, we can write
\begin{align}\label{model}
\lambda_i (\mathbf{p}^{(t)})=\mathbb{E}[y_{i}^{(t)} | \mathbf{p}^{(t)}= \mu_i(-\beta_{i,0} p_i^{(t)} +\langle\boldsymbol{\gamma}_{i,0},\mathbf{p}_{-i}^{(t)}\rangle)= \mu_i(\langle \boldsymbol{\theta}_{i,0},\mathbf{p}^{(t)}\rangle),\nonumber
\end{align}
where $\mu_i$ is \emph{known to the seller $i$} (but not to competitors $j \in [N]\setminus \{i\}$) and $\boldsymbol{\theta}_{i,0}$ is an $N$ dimensional vector \emph{unknown to everyone} with $i$-th entry equals to $-\beta_{i,0}$ and the rest of the entries are the values of $\boldsymbol{\gamma}_{i,0}$ ordered. The parameters $\beta_{i,0}$ and $\boldsymbol{\gamma}_{i,0}$ are commonly referred to as the $i$-th seller’s own-price sensitivity and the competitors’ price sensitivities, respectively. The demands $\{y_j^{(t)}\}_{j \in [N]}$ within a period $t$ can be correlated across sellers conditional on $\mathbf{p}^{(t)}$. The model in \eqref{eq:expon_model} generates samples $y_i^{(t)}$ that reflect not only each seller's own price $p_i^{(t)}$ but also the strategies of the competitors $\mathbf{p}^{(t)}_{-i} = (p_j^{(t)})_{j \in [N], j \neq i}$. The exponential model in \eqref{eq:expon_model} allows us to generalize linear models and also to accommodate both binary and continuous outcomes: when binary, it might represent the indicator of whether seller $i$ sold the item or not at price $p_i^{(t)}$; when continuous, it can be seen as observed mean demand value at price $\mathbf{p}^{(t)}$ plus a mean-zero error. We present two relevant examples.

\begin{example}[Binary Response Model]
We consider a generalization of the monopolistic binary response model considered by \citet{fan2024policy, bracale2025dynamic} (and corresponding references). Each customer arrives at time $t$ and, given the price vector posted by sellers $\mathbf{p}^{(t)}$, samples a valuation for each seller $i$, $v_i^{(t)} = \langle \widetilde{\boldsymbol{\theta}}_{i,0},\mathbf{p}^{(t)}\rangle +\delta_i^{(t)}$, where $\widetilde{\boldsymbol{\theta}}_{i,0}$ is an unknown parameter and $\delta_i^{(t)}$ are i.i.d. across $t$ with c.d.f. $F_i$. Here $\delta_i^{(t)}$ are independent across $t$ but might be correlated across $i$ conditional on $\mathbf{p}^{(t)}$. Let $p_i^{(t)}$ be the price posted by seller $i$. A purchase from $i$ happens if $y_i^{(t)}=\mathbb{I}(p_i^{(t)}\leq v_i^{(t)})$ is equal to $1$. Let $\mu_i=1-F_i$ and note that $\mathbb{E}[y_i^{(t)}| \mathbf{p}^{(t)}]=\mu_i(p_i^{(t)}-\langle\widetilde{\boldsymbol{\theta}}_{i,0},\mathbf{p}^{(t)}\rangle) = \mu_i(\langle \boldsymbol{\theta}_{i,0},\mathbf{p}^{(t)}\rangle)$, where $\boldsymbol{\theta}_{i,0}$ coincides with $\widetilde{\boldsymbol{\theta}}_{i,0}$ except for the $i$-th coordinate which instead is equal to $1-[\widetilde{\boldsymbol{\theta}}_{i,0}]_i$, where $[\widetilde{\boldsymbol{\theta}}_{i,0}]_i$ is the $i$-th coordinate of $\widetilde{\boldsymbol{\theta}}_{i,0}$. This model coincides with the canonical exponential family in \eqref{eq:expon_model} with $\mu_i = b_i'$.
\end{example}

\begin{example}[Linear Regression Model]
Here we describe the model by \citet{li2024lego}. The observed demand at time $t$ for seller $i$ is assumed to have the form $y_i^{(t)} = \alpha_{i,0}-\beta_{i,0}p_i^{(t)}+\langle \boldsymbol{\gamma}_{i,0}, \mathbf{p}_{-i}^{(t)}\rangle+\eta_i^{(t)}$, $\boldsymbol{\eta}^{(t)} \mid \mathbf{p}^{(t)} \sim \mathcal{N}\left(\boldsymbol{0}, \Sigma\left(\mathbf{p}^{(t)}\right)\right)$, with $\Sigma\left(\mathbf{p}^{(t)}\right)$ positive semidefinite, where $\boldsymbol{\eta}^{(t)}  = (\eta_i^{(t)} )_{i \in [N]}$. This model matches the canonical exponential family in \eqref{eq:expon_model} with $\mu_i (x) = b_i'(x) = \alpha_{i,0}+x$.
\end{example}

We now present the assumptions on $\mu_i$ and the parameters $\boldsymbol{\theta}_{i,0}$. We assume that the average demand $\lambda_i$ is non-negative and $\partial_{p_i}\lambda_i <0$ among all values of $\{\boldsymbol{\theta}_{i,0}\}_{i \in [N]}$; a similar assumption is found in \citet{birge2024interfere, li2024lego, bracale2025revenue}. The above conditions hold if $\mu_i,\mu_i'>0$ and $\beta_{i,0}>0$, which are explicitly stated in \cref{assumption_param} and subsequent \cref{assumption_mu}.

\begin{assumption}[Bounded and feasible parameter space]
\label{assumption_param}
The unknown vector $\boldsymbol{\theta}_{i,0} = (-\beta_{i,0},\boldsymbol{\gamma}_{i,0})\in \Theta_i \subset \mathbb{R}^{N}$ satisfies $0<\underline{\beta}_i \leq \beta_{i,0} \leq \bar{\beta}_i$, $\boldsymbol{\gamma}_{i,0} \in K_i\subset \mathbb{R}^{N-1}$ for some $K_i$ compact convex, so that $\Theta_i \triangleq [-\bar{\beta}_i,-\underline{\beta}_i] \times K_i$ is compact convex in $\mathbb{R}^{N}$. $\Theta_i$ is known to seller $i$.
\end{assumption} 
We denote by $B_{\theta_i}>0$ the $\ell^2$ upper bound on $\Theta_i$, that is $ \lVert \boldsymbol{\theta}_{i} \rVert_2 \leq B_{\theta_i}$ for all $\boldsymbol{\theta}_{i} \in \Theta_i$. As far as $\mu_i=b_i'$, we already know that it is non-decreasing because $b_i$ is convex. However, to derive concentration bounds, we require $\mu_i'$ to be bounded away from zero. Before making this explicit, we need to define 
\begin{equation*}
\textstyle \mathcal{U}_i \triangleq [u^{\min}_i,u^{\max}_i],\quad u_i^{\min} \triangleq \inf_{(\boldsymbol{\theta}_i,\mathbf{p}) \in \Theta_i \times \mathcal{P}}\boldsymbol{\theta}_i^{\top}\mathbf{p},\quad u_i^{\max} \triangleq \sup_{(\boldsymbol{\theta}_i,\mathbf{p}) \in \Theta_i \times \mathcal{P}} \boldsymbol{\theta}_i^{\top}\mathbf{p},
\end{equation*}
which is known to seller $i$ by \cref{assumption_param}.

\begin{assumption}[Smoothness of $\mu_i$]
\label{assumption_mu}
We assume that $\mu_i \in \mathrm{C}^2(\mathcal{U}_i)$ with 
\begin{equation*}
0\leq \mu_i(u)\leq B_{i}, \quad\mu_i'(u) \geq c_{\mu_i}, \quad \forall u \in \mathcal{U}_i,
\end{equation*}
and $\mu_i'(u) \leq L_{\mu_i}$, $\forall u \in \mathbb{R}$, for some $B_i,c_{\mu_i},L_{\mu_i}>0$. The function $\mu_i$ is known to seller $i$ but unknown to competitors $j \in [N]\setminus \{i\}$.
\end{assumption}
Requiring $\mu_i'$ to be bounded away from zero is standard: it prevents degeneracy in the curvature of the log-likelihood (equivalently, keeps the Fisher information well-conditioned) and underpins the concentration inequalities used in our analysis; similar assumptions can be found in  \citet{russac2019weighted, russac2020algorithms,
balabdaoui2019least, groeneboom2018current}. Let $B_i''>0$ be such that $|\mu_i''(u)|\le B_i''$ for all $u\in\mathcal U_i$, which exists because $\mu_i\in\mathrm C^2(\mathcal U_i)$.

We define the noise term as $
\eta_i^{(t)} \triangleq  y_i^{(t)} - \mu_i(\langle\mathbf{p}^{(t)}, \boldsymbol{\theta}_{i,0}\rangle)$, which satisfies the following \cref{lemma:subgaussian_noise}, used to derive concentration bounds in our analysis.

\begin{lemma}\label{lemma:subgaussian_noise}
Under \cref{assumption_mu}, for every $i \in [N]$, $\eta_i^{(t)}$ is $L_{\mu_i}$-subgaussian conditionally on $ \mathcal{H}_i^{(t-1)}$.
\end{lemma}
\begin{remark}[Impossibility to learn competitors’ models]
Before proceeding to the next section, we highlight a key distinction between our framework and the MARL models discussed in \cref{sec:Related-Literature}. In most MARL settings, the realized rewards of all agents -- in our setting, the realized demands of competing sellers -- are observed by every participant. By contrast, our model assumes that these realized demands remain private. This privacy has an important implication: no seller can infer the policy of their competitors. To see this, recall that the demand model for seller $j$ is $\mathbb{E}[y_j^{(t)}| \mathbf{p}^{(t)}] = \mu_j(\langle \boldsymbol{\theta}_{j,0} , \mathbf{p}^{(t)} \rangle)$ where $\boldsymbol{\theta}_{j,0}$ and $\mu_j$ are the parameter vector and link function of seller $j$, respectively, both unknown to competitors. At each time $t$, a competitor seller $i \neq j$ only observes $\mathcal{H}_i^{(t)} = \cup_{s \leq t} \{ (y_i^{(s)}, \mathbf{p}^{(s)}) \}$ but  does not have access to the demand realizations $\{y_j^{(s)}\}_{s \leq t}$ of seller $j$. Thus, seller $i$ has no information from which to infer their competitor parameters $(\boldsymbol{\theta}_{j,0}, \mu_j)$, and hence predict their pricing policy.
\end{remark}

\subsection{Regret and Nash equilibrium}
Each seller competes with a dynamic optimal sequence of prices in hindsight while assuming that the other sellers would not have responded differently if this sequence of prices had been offered. Under such a dynamic benchmark, the objective of each seller $i \in [N]$ is to minimize the following regret metric in hindsight:
\begin{align}\label{eq:R}
\textstyle  \mathrm{Reg}_i(T)\triangleq \mathbb{E}\sum_{t=1}^T \operatorname{rev}_i(\Gamma_i(\mathbf{p}_{-i}^{(t)})|\mathbf{p}_{-i}^{(t)})-\operatorname{rev}_i(p_i^{(t)}|\mathbf{p}_{-i}^{(t)}), \quad \operatorname{rev}_i(p_i | \mathbf{p}_{-i})&\triangleq   p_i\mu_i( \langle \mathbf{p},
\boldsymbol{\theta}_{i,0}\rangle),
\end{align}
where $\Gamma_i:\mathcal{P}_{-i}\rightarrow \mathcal{P}_i$, $\Gamma_i(\mathbf{p}_{-i})\triangleq \arg \max _{p_i\in \mathcal{P}_i} \, \operatorname{rev}_i(p_i | \mathbf{p}_{-i})$, $\forall \mathbf{p}_{-i} \in \mathcal{P}_{-i}$, is called the $i$-th seller's \emph{best-response (BR) map}. We call \emph{BR operator}  the map $\boldsymbol{\Gamma}: \mathcal{P} \rightarrow \mathcal{P}$, $\boldsymbol{\Gamma}(\mathbf{p})= (\Gamma_1(\mathbf{p}_{-1}),\dots, \Gamma_N(\mathbf{p}_{-N}))$. A NE is a price vector $\mathbf{p}^{\star}$, solution to the following fixed point equation:
\begin{equation}\label{eq:NASH}
\mathbf{p}^{\star}= \boldsymbol{\Gamma}(\mathbf{p}^{\star})= (\Gamma_1(\mathbf{p}^{\star}_{-1}),\dots, \Gamma_N(\mathbf{p}^{\star}_{-N})).
\end{equation}

\begin{assumption}\label{strong-concavity-rev}
For every $i \in [N]$,  $\operatorname{rev}_i(\cdot | \mathbf{p}_{-i})$ is strongly concave in $\mathcal{P}_i$, uniformly in $\mathbf{p}_{-i}\in \mathcal{P}_{-i}$.\footnote{There exists $\xi_i>0$ independent of $\mathbf{p}_{-i}$ such that for all $x,y\in\mathcal{P}_i$ and all
$\mathbf{p}_{-i}$, $\operatorname{rev}_i(y | \mathbf{p}_{-i})
-
\operatorname{rev}_i(x | \mathbf{p}_{-i})\leq
\big\langle \nabla_{x}\operatorname{rev}_i(x | \mathbf{p}_{-i}), y-x\big\rangle
-\frac{\xi_i}{2} \|y-x\|_2^2$.} Moreover, $\mathbf{\Gamma}$ is a contraction, that is $\|\boldsymbol{\Gamma}(\mathbf{p})-\boldsymbol{\Gamma}(\mathbf{p}')\|_2 \leq L_{\boldsymbol{\Gamma}}\|\mathbf{p}-\mathbf{p}'\|_{2}$ for all $\mathbf{p}, \mathbf{p}'\in \mathcal{P}$ and some $L_{\boldsymbol{\Gamma}}\in [0,1)$.
\end{assumption}

\begin{remark}[Strong concavity and contraction of the best-response operator]
This remark provides intuition for the role of \cref{strong-concavity-rev}. The strong concavity of $\operatorname{rev}(\cdot \mid \mathbf{p}_{-i})$ is introduced to control the individual regret. We bound the regret by relating it to the deviation of the joint price vector from a NE. In particular, strong concavity allows us to establish $\operatorname{Reg}_i(T) \lesssim \sum_{t=1}^T \mathbb{E}\|\mathbf{p}^{(t)} - \mathbf{p}^\star\|_2^2$, for some NE $\mathbf{p}^\star$ (see \cref{reg_NE_control} for details). To control this deviation, we use the fixed-point identity $\mathbf{p}^\star = \boldsymbol{\Gamma}(\mathbf{p}^\star)$ and repeatedly apply the triangle inequality. This yields $\|\mathbf{p}^{(t)} - \mathbf{p}^\star\|_2
\lesssim
\sum_{j=0}^{t-1} L_{\boldsymbol{\Gamma}}^j
\;+\;
L_{\boldsymbol{\Gamma}}^{t-1}\|\mathbf{p}^{(0)} - \mathbf{p}^\star\|_2$ (see \cref{reg_NE_control_2}) , where $L_{\boldsymbol{\Gamma}}$ is the Lipschitz constant of the best-response operator $\boldsymbol{\Gamma}$. For the cumulative error $\sum_{t=1}^T \mathbb{E}\|\mathbf{p}^{(t)} - \mathbf{p}^\star\|_2^2$ to remain controlled, it is necessary that $0 \le L_{\boldsymbol{\Gamma}} < 1$, i.e., that $\boldsymbol{\Gamma}$ is a contraction. This condition prevents errors from amplifying across iterations and ensures stability of the learning dynamics.
\end{remark}

The {\bf \emph{strong concavity}} condition in \cref{strong-concavity-rev} is standard in the literature  \citep{bracale2025revenue, scutari2014real, li2024lego, tsekrekos2024variational}.  It encompasses linear demand models \citep{guda2026demand, li2024lego}, concave demand specifications, and more generally $s_i$-concave demand functions with $s_i \geq 1$ (see \citet{bracale2025revenue}, Appendix I). Importantly, it guarantees that the best-response map $\Gamma_i(\cdot)$ is well defined, that is, it admits a unique solution for every $\mathbf{p}_{-i}$. This solution can be obtained by solving the first-order condition $\partial_{p_i}\operatorname{rev}_i(p_i \mid \mathbf{p}_{-i}) = 0$ and projecting the solution onto $\mathcal{P}_i$. 

The {\bf \emph{contraction property}} in \cref{strong-concavity-rev} is typically ensured by combining a shape constraint on $\mu_i$ with a restriction on the model parameters. For instance, \citet{bracale2025revenue} assumes that $\mu_i$ is $s_i$-concave for some $s_i > -1$, together with the condition $\sup_{i \in [N]} \|\boldsymbol{\gamma}_{i,0}\|_1 / \beta_{i,0} < 1$. Similar parameter restrictions appear in \citet{li2024lego} (Assumption 2) and \citet{kachani2007modeling}. 

Finally, strong concavity ensures the existence of a Nash equilibrium, while the contraction of the best-response operator guarantees its uniqueness. Existence follows directly from Theorem 3 in \citet{scutari2014real}, and uniqueness is an immediate consequence of the Banach contraction theorem.

\section{Algorithm} \label{sec:algo}

We build on the UCB framework of \citet{russac2020algorithms}, adapting it to our competitive setting with endogenous and partially observed interactions. Our extension constructs confidence bounds for the revenue functions via a penalized likelihood approach, while eliminating the need for a coordinated initial exploration phase.

\textbf{Penalized MLE.} By the exponential model in \eqref{eq:expon_model}, the (penalized) maximum likelihood estimator $\widehat{\boldsymbol{\theta}}_i^{(t)}$ based on the observed
$y_i^{(1)},\dots,y_i^{(t-1)}$ and the selected prices $\mathbf{p}^{(1)},\dots,\mathbf{p}^{(t-1)}$ is
defined as the maximizer of $\sum_{s=1}^{t-1} \ln( \mathbb{P}_{\boldsymbol{\theta}_i}(y_i^{(s)} | \mathbf{p}^{(s)})) - \frac{\lambda}{2} \lVert \boldsymbol{\theta}_i\rVert_2^2$, which equals
\begin{align}\label{eq_log_likelihood_SW}
\textstyle \sum_{s=1}^{t-1} y_i^{(s)} \langle \boldsymbol{\theta}_{i},  \mathbf{p}^{(s)}\rangle - b_i(\langle \boldsymbol{\theta}_{i} , \mathbf{p}^{(s)}\rangle) + c_i(y_i^{(s)})- \frac{\lambda}{2} \lVert \boldsymbol{\theta}_i\rVert_2^2
\end{align}
which is strongly concave because, by convexity of $b_i$, $\sum_{s=1}^{t-1} \ln( \mathbb{P}_{\boldsymbol{\theta}_i}(y_i^{(s)} | \mathbf{p}^{(s)}))$ is concave in $\boldsymbol{\theta}_i$. Define $\widehat{\boldsymbol{\theta}}_i^{(t)\nML}$ as the unique solution of
\begin{equation}
\label{eq_MLE_SW}
\textstyle \sum_{s=1}^{t-1} \{y_i^{(s)} - \mu_i(\langle \mathbf{p}^{(s)}, \boldsymbol{\theta}_i\rangle )\} \mathbf{p}^{(s)} - \lambda \boldsymbol{\theta}_i = \boldsymbol{0} ,
\end{equation}
where the LHS is the derivative of the penalized log-likelihood in \eqref{eq_log_likelihood_SW}. Following the notation in \citet{russac2020algorithms}, we introduce
\begin{equation}
\label{eq_Design_matrix_SW}
\textstyle V_{i}^{(t-1)} \triangleq \sum_{s=1}^{t-1} \mathbf{p}^{(s)} \mathbf{p}^{(s)\top} + \frac{\lambda}{c_{\mu_i}} I_N ,\quad \text{and} \quad g_i^{(t-1)}(\boldsymbol{\theta}_i) \triangleq \sum_{s= 1}^{t-1} \mu_i(\langle \mathbf{p}^{(s)}, \boldsymbol{\theta}_i \rangle) \mathbf{p}^{(s)} + \lambda \boldsymbol{\theta}_i.
\end{equation}
For fixed $i$ and $t$, $g_i^{(t-1)}$ denotes a surrogate function introduced to obtain an upper bound on the matrix $V_i^{(t-1)}$ -- in the sense of partial ordering on the non-negative matrices. This bound plays a key role in controlling the upper confidence bound that will be formally defined in \cref{def:UCB_conditioned}. At every time $t$, the $i$-th seller estimator is $\widetilde{\boldsymbol{\theta}}_i^{(t)\nML}$, defined by
\begin{equation}
\label{eq:theta_tilde_SW}
\textstyle \widetilde{\boldsymbol{\theta}}_i^{(t)\nML}  =  \operatorname{argmin}_{\lVert \boldsymbol{\theta}_i \rVert_2 \leq B_{\theta_i}}  \lVert g_i^{(t-1)}(\widehat{\boldsymbol{\theta}}_i^{(t)\nML}) - g_i^{(t-1)}(\boldsymbol{\theta}_i)
 \rVert_{V_{i}^{(t-1),-1}},
\end{equation} 
where $V_{i}^{(t-1),-1}$ is the inverse of $V_{i}^{(t-1)}$. We need to consider $\widetilde{\boldsymbol{\theta}}_i^{(t)\nML}$ because $\widehat{\boldsymbol{\theta}}_i^{(t)\nML}$ is not guaranteed to satisfy $\lVert \widehat{\boldsymbol{\theta}}_i^{(t)\nML} \rVert_2 \leq B_{\theta_i}$. Here, $\widetilde{\boldsymbol{\theta}}_i^{(t)\nML}$ should be understood as a ``projection''
onto the admissible parameter set.

\textbf{Pricing Policy via UCB.} In this section, we describe how each seller $i$ determines their prices $p_i^{(t)}$ at a given time $t$. Fix $t \in [T]$ and let $\widetilde{\boldsymbol{\theta}}_i^{(t)\nML}=(-\widetilde{\beta}_i^{(t)\nML},\widetilde{\boldsymbol{\gamma}}_i^{(t)\nML})$ be the parameter estimate defined in \eqref{eq:theta_tilde_SW}. For any $(p_i,\mathbf{p}_{-i})\in \mathcal{P}$, we define the estimated revenue as the plug-in estimator:
\begin{align*}
\widehat{ \operatorname{rev}}_i^{(t)}(p_i | \mathbf{p}_{-i}) =  p_i\mu_{i}( -\widetilde{\beta}_i^{(t)\nML}p_i+\langle\widetilde{\boldsymbol{\gamma}}_i^{(t)\nML}, \mathbf{p}_{-i}\rangle)=p_i\mu_i(\langle \mathbf{p},
\widetilde{\boldsymbol{\theta}}_i^{(t)\nML} \rangle).
\end{align*}
Following standard UCB-based approaches \citep{russac2020algorithms, russac2019weighted}, we define the upper confidence bound (UCB) of the estimated revenue as 
\begin{align}\label{def:UCB_conditioned}
\mathrm{UCB}_i^{(t)}(p_i| \mathbf{p}_{-i}^{(t-1)})= \widehat{ \operatorname{rev}}_i^{(t)}(p_i | \mathbf{p}_{-i}^{(t-1)}) +\sigma_i^{(t-1)}(p_i,\delta),
\end{align}
where for all $s = 0,1,\dots, T-1$ we define
\begin{align}\label{def:UCB}
&\textstyle \sigma_i^{(s)}(p_i,\delta) \triangleq \rho_i^{(s)}(\delta) p_i \lVert (p_i,\mathbf{p}_{-i}^{(s)}) \rVert_{V_i^{(s),-1}},\quad \rho_i^{(s)}(\delta) \triangleq \tfrac{2 L_{\mu_i}}{c_{\mu_i}}  \left( \mathrm{c}^{(s)}_i(\delta)
+ B_{\theta_i}\sqrt{c_{\mu_i} \lambda}  \right),\nonumber\\
&\mathrm{c}^{(s)}_i(\delta) \triangleq  L_{\mu_i} \sqrt{2 \log\left(\nicefrac{T}{\delta}\right) + N \log\left( 1 + \nicefrac{B_p^2 s}{N \lambda} \right) } .
\end{align}
We finally define $p_i^{(t)}= \widehat{\Gamma}_i^{(t)}(\mathbf{p}^{(t-1)}_{-i}) \in \operatorname{argmax}_{p_i \in \mathcal{P}_i} \, \mathrm{UCB}_i^{(t)}(p_i| \mathbf{p}_{-i}^{(t-1)})$, which represents an optimistic approximation of the best-response map $\Gamma_i$ at $\mathbf{p}^{(t-1)}_{-i}$. Using this notation, we are now ready to present our $\ML$ algorithm (Penalized Maximum Likelihood Generalized Linear Upper Confidence Bound), which is fully described in \cref{alg:S-GLM}.


\begin{algorithm}[h]
\caption{$\ML$}
   \label{alg:S-GLM}
\begin{algorithmic}
\STATE {\bfseries Input:} Probability $\delta$, regularization $\lambda$, 
upper bounds $B_p$, $B_{\theta_i}$, $B_i$, $c_{\mu_i}$, and $L_{\mu_i}$.
\STATE {\bfseries Initialize:} $\widehat{\boldsymbol{\theta}}_i^{(0)\nML} \in \Theta_i$, each seller starts with any $p_i^{(0)} \in \mathcal{P}_i$ and observes the corresponding $y_i^{(0)}$. At the end of time $t=0$ each seller $i$ observes $\mathbf{p}^{(0)}$ and sets $V_i^{(0)} = \mathbf{p}^{(0)}\mathbf{p}^{(0)\top} + \nicefrac{\lambda}{c_{\mu_i}} I_N$.
\FOR{$t=1$ {\bfseries to} $T$}
\STATE $\bullet$ Each seller $i$ computes $\widehat{\boldsymbol{\theta}}_i^{(t)\nML}$ according to \eqref{eq_MLE_SW}. If $\lVert \widehat{\boldsymbol{\theta}}_i^{(t)\nML} \rVert_2 \leq B_{\theta_i}$ let $\widetilde{\boldsymbol{\theta}}_i^{(t)\nML} =  \widehat{\boldsymbol{\theta}}_i^{(t)\nML}$. Else compute
   $\widetilde{\boldsymbol{\theta}}_i^{(t)\nML}$ defined in \eqref{eq:theta_tilde_SW}.
   \STATE $\bullet$ Each seller $i$ set $p_i^{(t)} \leftarrow \widehat{\Gamma}_i^{(t)}(\mathbf{p}^{(t-1)}_{-i})\in \operatorname{argmax}_{p_i \in \mathcal{P}_i} \, \mathrm{UCB}_i^{(t)}(p_i| \mathbf{p}_{-i}^{(t-1)})$.
  \STATE $\bullet$ Each seller $i$ observes $y_i^{(t)}$. 
  \STATE $\bullet$ Each seller $i$ observes $\mathbf{p}^{(t)}$.
  \STATE $\bullet$ Each seller $i$ updates $V_{i}^{(t)} \leftarrow V_{i}^{(t-1)}+ \mathbf{p}^{(t)} \mathbf{p}^{(t)\top}$.
   \ENDFOR
\end{algorithmic}
\end{algorithm}

\section{Regret Analysis} \label{sec:concentration}

In this section, we present our main result, \cref{theorem_regret_SW}, which shows the $\widetilde{O}(\sqrt{T})$ bound on the individual regret.


\begin{theorem}[Regret of $\ML$]
\label{theorem_regret_SW}
Suppose assumptions \eqref{assumption_param}, \eqref{assumption_mu},  and \eqref{strong-concavity-rev}. For every $i \in [N]$ we have that, if  $\delta= \frac{1}{T^\gamma}$ for some fixed $\gamma >1$, then, as long as $c_{\mu_j} B_p^2 \leq \lambda \leq O(N\log(T))$, $\operatorname{Reg}_i(T)\leq C + O(N^{2}\sqrt{T}\log(T))$, for some $C>0$ independent of $T$. An upper bound expression for $\operatorname{Reg}_i(T)$ depending on $\lambda,N$ and $T$ can be found in \eqref{eq:regret_dependence_on_lambda}. Additionally we have
\begin{equation}\label{eq:conv_sum_NE}
\textstyle  \sum_{t=1}^T\mathbb{E}[\|\mathbf{p}^{(t)}-\mathbf{p}^{\star}\|_2^2] = O(N^{2}\sqrt{T}\log(T)),
\end{equation}
where $\mathbf{p}^{(t)}=\widehat{\boldsymbol{\Gamma}}^{(t)}(\mathbf{p}^{(t-1)})=(\widehat{\Gamma}_i^{(t)}(\mathbf{p}_{-i}^{(t-1)}))_{i\in[N]}$, i.e. $p_i^{(t)} = \widehat{\Gamma}_i^{(t)}(\mathbf{p}_{-i}^{(t-1)})$, $\forall i \in [N]$.

\end{theorem}

\begin{proof}[{\bf Proof Sketch}]
Using Taylor expansion around $\Gamma_i(\mathbf{p}_{-i}^{(t)})$, we can write $\operatorname{rev}_i\big(\Gamma_i(\mathbf{p}_{-i}^{(t)}) | \mathbf{p}_{-i}^{(t)}\big)
-
\operatorname{rev}_i\big(p_i^{(t)} | \mathbf{p}_{-i}^{(t)}\big)\lesssim (\Gamma_i(\mathbf{p}_{-i}^{(t)})-p_i^{(t)})^2\leq \|\boldsymbol{\Gamma}(\mathbf{p}^{(t)})-\mathbf{p}^{(t)}\|^2_2$,
and summing over $t$ and taking $\mathbb{E}$ both sides we have
\begin{align}\label{eq:inst_reg_remark}
\textstyle\mathrm{Reg}_i(T)\lesssim \sum_{t=1}^T\mathbb{E} \|\boldsymbol{\Gamma}(\mathbf{p}^{(t)})-\mathbf{p}^{(t)}\|_2^2.
\end{align}
Adding and subtracting the fixed point $\mathbf{p}^{\star}$ from the RHS of \eqref{eq:inst_reg_remark} and using that $\mathbf{p}^{\star}=\boldsymbol{\Gamma}(\mathbf{p}^{\star})$ along with Lipschitzianity of $\boldsymbol{\Gamma}$ we get that $\mathrm{Reg}_i(T)$ is bounded by
\begin{align}\label{reg_NE_control}
\textstyle\sum_{t=1}^T\mathbb{E} \|\mathbf{p}^{(t)}-\mathbf{p}^{\star}\|_2^2 =\sum_{t=1}^T\mathbb{E} \|\widehat{\boldsymbol{\Gamma}}^{(t)}(\mathbf{p}^{(t-1)})-\mathbf{p}^{\star}\|_2^2.
\end{align}
Now let $\sigma^{(t)}(\delta)\triangleq \sum_{i \in [N]} \sigma_i^{(t)}(p_i^{(t+1)},\delta)$. Adding and subtracting $\boldsymbol{\Gamma}(\mathbf{p}^{(t-1)})$ from  $\|\widehat{\boldsymbol{\Gamma}}^{(t)}(\mathbf{p}^{(t-1)})-\mathbf{p}^{\star}\|_2$ and iterating this inequality $T$ times, we get
\begin{align}\label{reg_NE_control_2}
\textstyle \|\widehat{\boldsymbol{\Gamma}}^{(t)}(\mathbf{p}^{(t-1)})-\mathbf{p}^{\star}\|\lesssim \left(\sum_{j=0}^{t-1}L^{j}_{\boldsymbol{\Gamma}} \sqrt{\sigma^{(t-j-1)}(\delta)}\right)+L^{t-1}_{\boldsymbol{\Gamma}},
\end{align}
which holds with high probability by the following \cref{prop_SW_GLM_anytime_upper}.

\begin{proposition} 
\label{prop_SW_GLM_anytime_upper}
Let the assumptions of the theorem hold. For every $0<\delta<1$, the event $\cap_{t \in [T]}\{\|\widehat{\boldsymbol{\Gamma}}^{(t)}(\mathbf{p}^{(t-1)})-\boldsymbol{\Gamma}(\mathbf{p}^{(t-1)})\|_2^2 \lesssim \sigma^{(t-1)}(\delta)\}$, holds with probability at least $1-2N\delta$.
\end{proposition}
Plugging the bound from \eqref{reg_NE_control_2}
into the expectation \eqref{reg_NE_control} (and ignoring the expectation over the complementary probability set, which can be shown to be negligible), and using that $L_{\boldsymbol{\Gamma}}<1$, we obtain $\textstyle\mathrm{Reg}_i(T)\lesssim \mathbb{E}\sum_{t=0}^{T-1} \sigma^{(t)}(\delta)$. Now, for simplicity of the proof's sketch, consider $\lambda$ a fixed constant and choose $\delta = 1/T^2$. From the definition of $\rho_i^{(t)}(\delta)$ in \cref{def:UCB} it is easy to see that  $\rho_i^{(t)}(\delta)\lesssim \sqrt{N\log(T)}$, hence 
\begin{align*}
\textstyle \sigma^{(t)}(\delta)=\sum_{j \in [N]} \rho_j^{(t)}(\delta)p_j^{(t+1)} \| (p_j^{(t+1)},\mathbf{p}_{-j}^{(t)} )\|_{V_{j}^{(t),-1}}\lesssim \sqrt{N\log(T)}\sum_{j \in [N]} \| (p_j^{(t+1)},\mathbf{p}_{-j}^{(t)} )\|_{V_{j}^{(t),-1}}.
\end{align*}
Therefore we obtain $\mathrm{Reg}_i(T) \lesssim \sqrt{N \log(T)} \cdot K(T)$, where $K(T) = \mathbb{E}\sum_{t=0}^{T-1} \sum_{j \in [N]} \| (p_j^{(t+1)},\mathbf{p}_{-j}^{(t)} )\|_{V_{j}^{(t),-1}}$.
The reader can immediately recognize the correspondence of $K(T)$ with the elliptical potential lemma used in the bandit literature, with the only difference being that, in a multi-agent setting, the presence of $p_j^{(t+1)}$ does not allow a direct application of this lemma. However, in \cref{lemma:Elliptical-Lemma} we prove that it is still possible to bound $K(T)$ by $O(N^{3/2}\sqrt{T\log(T)})$ without additional assumptions. This completes the regret convergence since we obtain $\mathrm{Reg}_i(T) \lesssim \sqrt{N \log(T)} \cdot K(T) = O(N^{2}\sqrt{T}\log(T))$.
The convergence of $\sum_{t=1}^T\mathbb{E} \|\mathbf{p}^{(t)}-\mathbf{p}^{\star}\|_2^2$ comes from \eqref{reg_NE_control}.
\end{proof}

\begin{remark}[Optimality of the $\ML$ algorithm]
Our rate aligns with the linear-case upper bound in \citet{li2024lego} for both $T$ and $N$. The algorithm of \citet{li2024lego} hits the minimax-optimal $\widetilde{\Theta}(\sqrt{T})$ regret under linear demand. Because our framework strictly generalizes the linear case and $\ML$ still achieves $\widetilde{\mathcal O}(\sqrt{T})$ regret, our mechanism is order-optimal -- minimax optimal up to logarithmic factors -- in this more general setting. However, it remains unclear whether the $N^2$ dependence is optimal, specifically, whether the individual regret is indeed $\Omega(N^2 \sqrt{T})$. Proving such a result would require a deeper analysis, and we view this as an interesting and worthwhile direction for future research.
\end{remark}

\begin{remark}[Computational complexity.]\label{remark:comput-complexity}
At round $t$, seller $i$ updates the matrix $V_i^{(t)}=V_i^{(t-1)}+\mathbf p^{(t)}\mathbf p^{(t)\top}$ and its inverse. Using the Sherman--Morrison formula to update $V_i^{(t),-1}$ costs $O(N^2)$ time (and forming $V_i^{(t)}$ itself is also $O(N^2)$). To evaluate the UCB bonus for a candidate price $p_i$, one needs $\|(p_i,\mathbf p_{-i})\|_{V_i^{(t),-1}}$, which requires a matrix--vector product $V_i^{(t),-1}\mathbf p$ in $O(N^2)$ time and then a dot product in $O(N)$. Both the true best response and the optimistic best response solve a one-dimensional concave maximization problem. In practice, this can be implemented via a 1-D convex optimization (e.g., bisection or Newton), which is costless, i.e., $O(1)$ (this implies that the UCB optimization is tractable and does not suffer from the NP-hardness of high-dimensional problems). Over $T$ rounds, the time is $O(TN^2)$ per seller, and $O(TN^3)$ if all $N$ sellers are run. Space usage per seller is $O(N^2)$ to store $V_i^{(t)}$ and (optionally) $V_i^{(t),-1}$ plus $O(N)$ for current vectors; if past histories are not stored, $O(N^2)$ dominates.
\end{remark}

\textbf{Convergence to Nash Equilibrium.} The bound in \eqref{eq:conv_sum_NE} alone does not imply a convergence rate for $\mathbb{E} \left[\|\mathbf{p}^{(T)}-\mathbf{p}^{\star}\|_2^2\right]$ as $T\to\infty$. Indeed, a bound on the partial sums only yields pointwise convergence if the sequence is sufficiently regular (e.g., eventually monotone), or under an initial i.i.d. pricing exploration phase, like in \citet{li2024lego} and \citet{bracale2025revenue}, which ensures that the estimation error satisfies $\|\widehat{\boldsymbol\theta}_i^{(T)}-\boldsymbol\theta_{i,0}\|_2=O_P(T^{-1/2})$ under standard regularity conditions, which in turn yields the parametric-rate convergence $\mathbb{E} \left[\|\mathbf{p}^{(T)}-\mathbf{p}^{\star}\|_2^2\right] =O(T^{-1/2})$. In our setting, however, prices $\{\mathbf{p}^{(t)}\}_{t\le T}$ are generated \emph{adaptively} (their selection depends on past noise realizations), and we do not impose an i.i.d.\ excitation phase. In the absence of this phase, the estimation error can typically be controlled only along directions that are sufficiently excited by the played actions. As a result, $\|\widehat{\boldsymbol\theta}_i^{(T)}-\boldsymbol\theta_{i,0}\|_2$ does not generally decay at the parametric rate $O_P(T^{-1/2})$, and one should not expect a pointwise convergence rate for $\mathbb{E} \left[\|\mathbf{p}^{(T)}-\mathbf{p}^{\star}\|_2^2\right]$ based solely on
\eqref{eq:conv_sum_NE}. Nevertheless, as we show in \cref{theorem_regret_SW}, our primary objective is regret optimality, which can be achieved even without strong pointwise convergence of the price sequence to a Nash equilibrium. This is consistent with the broader bandit literature, where regret-optimal algorithms often rely on adaptive (non-i.i.d.) action selection, as in OFU/UCB-style methods for linear and generalized linear bandits \citep{dani2008stochastic,abbasi2011improved,filippi2010parametric,russac2019weighted}.



\section{Limitations and future directions.} \label{sec:conclusion}

We assume a \emph{known} increasing link $\mu_i$, strong concavity of per-round revenues, and a contraction-type condition that guarantees a unique NE. None of these is new: they are standard in dynamic pricing and learning-in-games and \emph{strictly include} the linear-demand setting of \citet{li2024lego} as a special case. These conditions ensure well-posed best responses and well-conditioned information for concentration, but they also circumscribe our scope to stationary parameters and synchronous price moves under a unique-equilibrium regime.

A natural avenue is to relax link knowledge by replacing it with shape constraints (e.g., monotonicity or $s$-concavity) and to derive calibration-free confidence bonuses that adapt to local curvature while remaining free of any initial exploration phase. On the game-theoretic side, moving beyond global contractions -- to allow multiple or set-valued equilibria -- suggests a variational-inequality/monotone-operator treatment to obtain no-regret convergence. Algorithmically, Bayesian or Thompson sampling approaches to GLMs under our information structure (where rivals’ revenues remain unobserved) are particularly appealing, as they tend to yield less conservative confidence bounds in practice compared to UCB-based methods. Finally, incorporating nonstationarity in $\boldsymbol{\theta}_{i,0}$, inventory constraints, and asynchronous observability of rivals’ prices would bring the model even closer to more real market scenarios.



\bibliography{main}
\bibliographystyle{apalike}

\newpage
\appendix

\section*{\centering \Large Appendix of \\
Online Price Competition under Generalized Linear Demands
}
\appendix

\startcontents[appendix]

\section*{Appendix Contents}
\printcontents[appendix]{}{1}{}

\newpage

\section{Proof of Proposition \ref{prop_SW_GLM_anytime_upper}}
\label{subsec:prop_corollary_prop_SW_GLM_anytime}

We want to prove that, under the assumptions of \eqref{theorem_regret_SW}, for every $0<\delta<1$,
$$
\mathbb{P}\left(\cap_{t \in [T]}\left\{\|\widehat{\boldsymbol{\Gamma}}^{(t)}(\mathbf{p}^{(t-1)})-\boldsymbol{\Gamma}(\mathbf{p}^{(t-1)})\|_2^2 \leq C_2 \sigma^{(t-1)}(\delta)\right\}\right)\geq 1-2N\delta,
$$
where $\sigma^{(t)}(\delta)\triangleq \sum_{i \in [N]} \sigma_i^{(t)}(p_i^{(t+1)},\delta)$ and $C_2 = \frac{4}{\min_{i\in [N]}\xi_i}$, with $\xi_i>0$ such that for all $x,y\in\mathcal{P}_i$ and all
$\mathbf{p}_{-i}$, 
$$
\operatorname{rev}_i(y | \mathbf{p}_{-i})
-
\operatorname{rev}_i(x | \mathbf{p}_{-i})\leq
\big\langle \nabla_{x}\operatorname{rev}_i(x | \mathbf{p}_{-i}), y-x\big\rangle
-\frac{\xi_i}{2} \|y-x\|^2
$$
For every $i \in [N]$ the constant $\xi_i$ is independent of $\mathbf{p}_{-i}$ by Assumption \eqref{strong-concavity-rev}. We recall the definitions
$$
\operatorname{rev}_i(p_i | \mathbf{p}_{-i}) = p_i\mu_i( -\beta_{i,0}p_i+\langle \boldsymbol{\gamma}_{i,0},\mathbf{p}_{-i}\rangle), \qquad \forall (p_i , \mathbf{p}_{-i}) \in \mathcal{P},
$$
and
$$
\widehat{ \operatorname{rev}}_i^{(t)}(p_i | \mathbf{p}_{-i}) =  p_i\mu_{i}( -\widetilde{\beta}_i^{(t)\nML}p_i+\langle\widetilde{\boldsymbol{\gamma}}_i^{(t)\nML}, \mathbf{p}_{-i}\rangle), \qquad \forall (p_i , \mathbf{p}_{-i}) \in \mathcal{P}.
$$
We also define
$$
\textstyle p_i^{(t)\star} =\Gamma_i(\mathbf{p}^{(t)}_{-i})=  \underset{p_i \in \mathcal{P}_i}{\operatorname{argmax}} \, \operatorname{rev}_i(p_i | \mathbf{p}_{-i}^{(t)}) , \qquad \forall  \mathbf{p}_{-i} ^{(t)}\in \mathcal{P}_{-i},
$$
and 
$$
\textstyle p_i^{(t+1)} =\widehat{\Gamma}^{(t+1)}_i(\mathbf{p}^{(t)}_{-i}) \in  \underset{p_i \in \mathcal{P}_i}{\operatorname{argmax}} \, \text{UCB}_i^{(t+1)}(\cdot |\mathbf{p}_{-i}^{(t)}) ,\qquad \forall  \mathbf{p}_{-i} ^{(t)}\in \mathcal{P}_{-i},
$$
where we recall from \cref{def:UCB_conditioned} that
\begin{align*}
\mathrm{UCB}_i^{(t+1)}(p_i| \mathbf{p}_{-i}^{(t)})= \widehat{ \operatorname{rev}}_i^{(t+1)}(p_i | \mathbf{p}_{-i}^{(t)}) +\sigma_i^{(t)}(p_i,\delta).
\end{align*}
For the first part of the proof, we follow similar steps to \citet{russac2020algorithms}, Corollary 1. We start by upper-bounding the following difference:
\begin{align*}
\operatorname{rev}_i(p_i^{(t)\star} | \mathbf{p}_{-i}^{(t)}) - \operatorname{rev}_i(p_i^{(t+1)} | \mathbf{p}_{-i}^{(t)})
&= \underbrace{\operatorname{rev}_i(p_i^{(t)\star} | \mathbf{p}_{-i}^{(t)})-\widehat{\operatorname{rev}}^{(t+1)}_i(p_i^{(t)\star} | \mathbf{p}_{-i}^{(t)})}_{A1} \\
&+\underbrace{\widehat{\operatorname{rev}}^{(t+1)}_i(p_i^{(t)\star} | \mathbf{p}_{-i}^{(t)})-  \widehat{\operatorname{rev}}^{(t+1)}_i(p_i^{(t+1)} | \mathbf{p}_{-i}^{(t)})}_{A2}\\
&+  \underbrace{\widehat{\operatorname{rev}}^{(t+1)}_i(p_i^{(t+1)} | \mathbf{p}_{-i}^{(t)})-\operatorname{rev}_i(p_i^{(t+1)} | \mathbf{p}_{-i}^{(t)})}_{A3}.
\end{align*}

To upper bound $A1$ and $A3$ we use the following \cref{prop_SW_GLM_concentration}, whose proof can be found in \cref{subsection:prop_SW_GLM_concentration}.
\begin{lemma}
\label{prop_SW_GLM_concentration}
Let assumptions \eqref{assumption_param}, \eqref{assumption_mu} and \eqref{strong-concavity-rev} hold. Recall that for every $t=0,1,\dots, T-1$ we have defined
\begin{align*}
&\sigma_i^{(t)}(p_i,\delta) = \rho_i^{(t)}(\delta) p_i \lVert (p_i,\mathbf{p}_{-i}^{(t)}) \rVert_{V_i^{(t),-1}},\\
&\rho_i^{(t)}(\delta) = \frac{2 L_{\mu_i}}{c_{\mu_i}}  \bigg( \mathrm{c}^{(t)}_i(\delta)
+ B_{\theta_i}\sqrt{c_{\mu_i} \lambda}  \bigg),\quad \text{and} \quad \mathrm{c}^{(t)}_i(\delta) =  L_{\mu_i} \sqrt{2 \log\left(\frac{T}{\delta}\right) + N \log\left( 1 + \frac{B_p^2 t}{N \lambda} \right) }.
\end{align*}

Let $0<\delta <1$ and 
$t \in [T]$. Let $\mathbf{p}$ be any 
$\mathcal{P}$-valued (possibly random) price vector. Then, simultaneously for all $t \in [T]$,
\begin{align*}
|\operatorname{rev}_i(p_i|\mathbf{p}_{-i})- \widehat{\operatorname{rev}}^{(t)}_i(p_i|\mathbf{p}_{-i})|= \big|p_i\mu_i(\langle\mathbf{p}, \boldsymbol{\theta}_{i,0}\rangle) - p_i\mu_i(\langle\mathbf{p} ,
\widetilde{\boldsymbol{\theta}}_i^{(t)\nML}\rangle) \big|\leq \rho_i^{(t-1)}(\delta)p_i\lVert \mathbf{p} \rVert_{V_{i}^{(t-1),-1}}
\end{align*}
holds with probability higher than $1- \delta$.
\end{lemma}
Thanks to Lemma \ref{prop_SW_GLM_concentration}, we can give an upper bound for the term $A1$ and for the term $A3$. With a union bound, we can simultaneously upper bound $A1$ and $A3$ for all $t \in [T]$ and the following holds

\begin{equation}
\label{eq:simult_upper_bound}
\mathbb{P}\left(\forall t \in [T], A1
 \leq \rho_i^{(t)}(\delta) p_i^{(t)\star}\| (p_i^{(t)\star},\mathbf{p}_{-i}^{(t)} )\|_{V_{i}^{(t),-1}}  \cap 
   A3
\leq \rho_i^{(t)}(\delta)p_i^{(t+1)}\| (p_i^{(t+1)},\mathbf{p}_{-i}^{(t)} )\|_{V_{i}^{(t),-1}} \right) \geq 1- 2\delta
\end{equation}

Let $E$ be the event in \eqref{eq:simult_upper_bound}. We now upper-bound $A2$.

\begin{align*}
A2 &= \widehat{\operatorname{rev}}^{(t+1)}_i(p_i^{(t)\star} | \mathbf{p}_{-i}^{(t)})-  \widehat{\operatorname{rev}}^{(t+1)}_i(p_i^{(t+1)} | \mathbf{p}_{-i}^{(t)}) \\
&= \underset{=\mathrm{UCB}_i^{(t+1)}(p_i^{(t)\star}| \mathbf{p}_{-i}^{(t)}) }{\underbrace{\widehat{\operatorname{rev}}^{(t+1)}_i(p_i^{(t)\star} | \mathbf{p}_{-i}^{(t)}) + \rho_i^{(t)}(\delta) p_i^{(t)\star}\| (p_i^{(t)\star},\mathbf{p}_{-i}^{(t)} )\|_{V_{i}^{(t),-1}}}} \\
& \qquad - 
\rho_i^{(t)}(\delta) p_i^{(t)\star}\| (p_i^{(t)\star},\mathbf{p}_{-i}^{(t)} )\|_{V_{i}^{(t),-1}}+\rho_i^{(t)}(\delta) p_i^{(t+1)} \lVert (p_i^{(t+1)},\mathbf{p}_{-i}^{(t)}) \rVert_{V_i^{(t),-1}} \\
& \qquad  \underset{=-\mathrm{UCB}_i^{(t+1)}(p_i^{(t+1)}| \mathbf{p}_{-i}^{(t)}) }{\underbrace{-\rho_i^{(t)}(\delta) p_i^{(t+1)} \lVert (p_i^{(t+1)},\mathbf{p}_{-i}^{(t)}) \rVert_{V_i^{(t),-1}}- \widehat{\operatorname{rev}}^{(t+1)}_i(p_i^{(t+1)} | \mathbf{p}_{-i}^{(t)})}} \\
&= \underset{\leq 0 \text{ by def of }p_i^{(t+1)} }{\underbrace{\mathrm{UCB}_i^{(t+1)}(p_i^{(t)\star}| \mathbf{p}_{-i}^{(t)}) -\mathrm{UCB}_i^{(t+1)}(p_i^{(t+1)}| \mathbf{p}_{-i}^{(t)})}}\\
& \quad -
\rho_i^{(t)}(\delta) p_i^{(t)\star}\| (p_i^{(t)\star},\mathbf{p}_{-i}^{(t)} )\|_{V_{i}^{(t),-1}}+\rho_i^{(t)}(\delta) p_i^{(t+1)} \lVert (p_i^{(t+1)},\mathbf{p}_{-i}^{(t)}) \rVert_{V_i^{(t),-1}}\\
&\leq \rho_i^{(t)}(\delta) p_i^{(t+1)}\| (p_i^{(t+1)},\mathbf{p}_{-i}^{(t)} )\|_{V_{i}^{(t),-1}}
- 
\rho_i^{(t)}(\delta) p_i^{(t)\star}\| (p_i^{(t)\star},\mathbf{p}_{-i}^{(t)} )\|_{V_{i}^{(t),-1}}.
\end{align*}

Now we put $A1,A2,A3$ together. Under the event $E$ (which occurs with a probability higher than $1-2\delta$), we have
\begin{align*}
\operatorname{rev}_i(p_i^{(t)\star} | \mathbf{p}_{-i}^{(t)}) - \operatorname{rev}_i(p_i^{(t+1)} | \mathbf{p}_{-i}^{(t)})
&\leq \underbrace{\rho_i^{(t)}(\delta)p_i^{(t)\star} \| (p_i^{(t)\star},\mathbf{p}_{-i}^{(t)} )\|_{V_{i}^{(t),-1}}}_{\text{coming from} A1} \\
& \qquad + 
\underbrace{\rho_i^{(t)}(\delta) p_i^{(t+1)}\| (p_i^{(t+1)},\mathbf{p}_{-i}^{(t)} )\|_{V_{i}^{(t),-1}}
- 
\rho_i^{(t)}(\delta) p_i^{(t)\star}\| (p_i^{(t)\star},\mathbf{p}_{-i}^{(t)} )\|_{V_{i}^{(t),-1}}}_{\text{coming from} A2} \\
& \qquad 
+ \underbrace{\rho_i^{(t)}(\delta) p_i^{(t+1)}\| (p_i^{(t+1)},\mathbf{p}_{-i}^{(t)} )\|_{V_{i}^{(t),-1}}}_{\text{coming from} A3} \\
& \leq 2 \rho_i^{(t)}(\delta) p_i^{(t+1)}\| (p_i^{(t+1)},\mathbf{p}_{-i}^{(t)} )\|_{V_{i}^{(t),-1}} = 2 \sigma_i^{(t)}(p_i^{(t+1)},\delta).
\end{align*}

Up to here we proved that for $0<\delta<1$ and for every $i \in [N]$,
\begin{align}\label{ineq:high_prob_rev}
\mathbb{P}\left( \cap_{t \in [T]}\left\{\operatorname{rev}_i(p_i^{(t)\star} | \mathbf{p}_{-i}^{(t)}) - \operatorname{rev}_i(p_i^{(t+1)} | \mathbf{p}_{-i}^{(t)}) \leq 2 \sigma_i^{(t)}(p_i^{(t+1)},\delta) \right\} \right) \geq 1-2\delta.
\end{align}
Now, by \cref{strong-concavity-rev} there exists $\xi_i>0$ independent of $\mathbf{p}_{-i}^{(t)}$ such that for all $x,y\in\mathcal{P}_i$ and all $\mathbf{p}_{-i}\in \mathcal{P}_{-i}$, 
$$
\operatorname{rev}_i(y | \mathbf{p}_{-i}^{(t)})
-
\operatorname{rev}_i(x | \mathbf{p}_{-i}^{(t)})\leq
\big\langle \nabla_{x}\operatorname{rev}_i(x | \mathbf{p}_{-i}^{(t)}), y-x\big\rangle
-\frac{\xi_i}{2} \|y-x\|^2.
$$
replacing $x$ with $p_i^{(t)*}$ and $y$ with $p_i^{(t+1)}$, together with the optimality of $p_i^{(t)*}$, we recover that 
\begin{equation}\label{ineq:strong_conc_star}
(p_i^{(t)*}-p_i^{(t+1)})^2 \leq \frac{2}{\xi_i}[\operatorname{rev}_i(p_i^{(t)\star} | \mathbf{p}_{-i}^{(t)}) - \operatorname{rev}_i(p_i^{(t+1)} | \mathbf{p}_{-i}^{(t)})]
\end{equation}
Now define $C_{i,2} \triangleq \frac{4}{\xi_i}$, and the event
$$
\mathcal{B}_i \triangleq \cap_{t =1}^T \left\{(\Gamma_i(\mathbf{p}^{(t)}_{-i})-\widehat{\Gamma}_i^{(t+1)}(\mathbf{p}^{(t)}_{-i}) )^{2}\leq C_{i,2}  \sigma_i^{(t)}(p_i^{(t+1)},\delta)\right\}.
$$
Since $(\Gamma_i(\mathbf{p}^{(t)}_{-i})-\widehat{\Gamma}_i^{(t+1)}(\mathbf{p}^{(t)}_{-i}))^{2}=(p_i^{(t)*}-p_i^{(t+1)})^2$, then \eqref{ineq:high_prob_rev} and \eqref{ineq:strong_conc_star} imply that for $\delta \in (0,1)$ and for every $i \in [N]$, $\mathbb{P}(\mathcal{B}_i)\geq 1-2\delta$. Now define $\mathcal{B} = \cap_{i \in [N]} \mathcal{B}_i$. Note that $\mathcal{B}$ holds with probability at least $1-2N\delta$:
$$
\mathbb{P}(\mathcal{B}^c) = \mathbb{P}(\cup_{i \in [N]}\mathcal{B}_i^c)\leq \sum_{i \in [N]}\mathbb{P}(\mathcal{B}_i^c) = \sum_{i \in [N]}(1-\mathbb{P}(\mathcal{B}_i))\leq \sum_{i \in [N]}2\delta = 2N\delta.
$$
Now define
\begin{equation*}
C_2 \triangleq \max_{i \in [N]}\{C_{i,2}\} =\frac{4}{\min_{i\in[N] }\xi_i}, \qquad \sigma^{(t)}(\delta)\triangleq \sum_{i \in [N]} \sigma_i^{(t)}(p_i^{(t+1)},\delta), \qquad \forall t \geq 1.
\end{equation*}
In $\mathcal{B}$ we have that
\begin{align*}
\|\widehat{\boldsymbol{\Gamma}}^{(t)}(\mathbf{p}^{(t-1)})-\boldsymbol{\Gamma}(\mathbf{p}^{(t-1)})\|_2 &= \sqrt{\sum_{i\in [N]}(\widehat{\Gamma}^{(t)}_i(\mathbf{p}^{(t-1)}_{-i})-\Gamma_i(\mathbf{p}^{(t-1)}_{-i}))^2}\\
&\leq \sqrt{\max_{i \in [N]}\{C_{i,2}\}}\sqrt{\sum_{i \in [N]}\sigma_i^{(t-1)}(\delta)} \\
&= \sqrt{C_2 \sigma^{(t-1)}(\delta)}.
\end{align*}
Let 
$$
\mathcal{G}=\cap_{t \in [T]}\left\{\|\widehat{\boldsymbol{\Gamma}}^{(t)}(\mathbf{p}^{(t-1)})-\boldsymbol{\Gamma}(\mathbf{p}^{(t-1)})\|_2^2 \leq C_2 \sigma^{(t-1)}(\delta)\right\}.
$$
Since $\mathcal{B}$ implies $\mathcal{G}$, we have

$$
\mathbb{P}\left(\mathcal{G}\right)\geq \mathbb{P}\left(\mathcal{B}\right)\geq 1-2N\delta.
$$
This completes the proof.

\section{Proof of Lemma \ref{prop_SW_GLM_concentration}}
\label{subsection:prop_SW_GLM_concentration}
For simplicity of notation, we write $\widehat{\boldsymbol{\theta}}_i^{(t)}$ for $\widehat{\boldsymbol{\theta}}_i^{(t)\nML}$ and
$\widetilde{\boldsymbol{\theta}}_i^{(t)}$ for $\widetilde{\boldsymbol{\theta}}_i^{(t)\nML}$. This proof follows a similar structure to the proof of Proposition 1 in \citet{russac2020algorithms}. Recall the definitions
$$
\operatorname{rev}_i(p_i | \mathbf{p}_{-i}) = p_i\mu_i( -\beta_{i,0}p_i+\langle \boldsymbol{\gamma}_{i,0},\mathbf{p}_{-i}\rangle)=p_i\mu_i( \langle \mathbf{p},
\boldsymbol{\theta}_{i,0}\rangle),
$$
and
$$
\widehat{ \operatorname{rev}}_i^{(t)}(p_i | \mathbf{p}_{-i}) =  p_i\mu_{i}( -\widetilde{\beta}_i^{(t)\nML}p_i+\langle\widetilde{\boldsymbol{\gamma}}_i^{(t)\nML}, \mathbf{p}_{-i}\rangle)=p_i\mu_i(\langle \mathbf{p},
\widetilde{\boldsymbol{\theta}}_i^{(t)} \rangle).
$$
We recall the definition $g_i^{(t-1)}: \mathbb{R}^d \mapsto \mathbb{R}^d$:
$$
g_i^{(t-1)}(\boldsymbol{\theta}_i) = 
\sum_{s= 1}^{t-1} \mu_i(\mathbf{p}^{(s)\top} \boldsymbol{\theta}_i) \mathbf{p}^{(s)} + \lambda \boldsymbol{\theta}_i.
$$
Let $J_i^{(t-1)}$ denote the Jacobian matrix of $g_i^{(t-1)}$.  We have
$$
J_i^{(t-1)}(\boldsymbol{\theta}_i) = \sum_{s=1}^{t-1} \mu_i^{\prime}(\mathbf{p}^{(s)\top} \boldsymbol{\theta}_i) \mathbf{p}^{(s)} \mathbf{p}^{(s)\top} + \lambda I_N.
$$
Thanks to the definition of the estimator $\widehat{\boldsymbol{\theta}}_i^{(t)\nML}$ defined in \cref{eq_MLE_SW}, we have 
$$
g_i^{(t-1)}(\widehat{\boldsymbol{\theta}}_i^{(t)\nML}) = \sum_{s= 1}^{t-1} \mathbf{p}^{(s)} y_i^{(s)}.
$$
We also introduce the martingale
$$
S_i^{(t-1)}=\sum_{s=1}^{t-1} \mathbf{p}^{(s)} \eta_i^{(s)}.
$$
We define the $G_i^{(t-1)}(\boldsymbol{\theta}_{i,0},\widetilde{\boldsymbol{\theta}}_i^{(t)})$ matrix as follows,
$$
G_i^{(t-1)}(\boldsymbol{\theta}_{i,0},\widetilde{\boldsymbol{\theta}}_i^{(t)}) = \int_{0}^1 J_i^{(t-1)}( u \boldsymbol{\theta}_{i,0} + (1-u) \widetilde{\boldsymbol{\theta}}_i^{(t)})   du .
$$
The Fundamental Theorem of Calculus gives
\begin{equation}
\label{eq_g_t_G_t} 
g_i^{(t-1)}(\boldsymbol{\theta}_{i,0}) -  g_i^{(t-1)}(\widetilde{\boldsymbol{\theta}}_i^{(t)})  = G_i^{(t-1)}(\boldsymbol{\theta}_{i,0},\widetilde{\boldsymbol{\theta}}_i^{(t)})(\boldsymbol{\theta}_{i,0}- \widetilde{\boldsymbol{\theta}}_i^{(t)}) .
\end{equation}
Knowing that both $\boldsymbol{\theta}_{i,0}$ and $\widetilde{\boldsymbol{\theta}}_i^{(t)}$ have an L2-norm smaller than
$B_{\theta_i}$, $\forall u \in [0,1],  \lVert u \boldsymbol{\theta}_{i,0} + (1-u) \widetilde{\boldsymbol{\theta}}_i^{(t)}) \rVert_2 \leq B_{\theta_i}$.
This implies in particular that
\begin{equation}
\label{eq_G_t}
G_i^{(t-1)}(\boldsymbol{\theta}_{i,0},\widetilde{\boldsymbol{\theta}}_i^{(t)}) \geq c_{\mu_i} \left( \sum_{s= 1}^{t-1} \mathbf{p}^{(s)} \mathbf{p}^{(s)\top} + \frac{\lambda} {c_{\mu_i}} I_N  \right) = c_{\mu_i} V_{i}^{(t-1)} ,
\end{equation}
which in turn ensures $G_i^{(t-1)}(\boldsymbol{\theta}_{i,0},\widetilde{\boldsymbol{\theta}}_i^{(t)})$ is invertible. Let $\mathbf{p}$ be any price vector in $\mathcal{P}$ (possibly random) and $t$ be a fixed time instant,

\begin{align}
|\operatorname{rev}_i(p_i | \mathbf{p}_{-i}) - \widehat{ \operatorname{rev}}_i^{(t)}(p_i | \mathbf{p}_{-i})| \leq p_i |\mu_i( \langle \mathbf{p},
\boldsymbol{\theta}_{i,0}\rangle) - \mu_i(\langle \mathbf{p},
\widetilde{\boldsymbol{\theta}}_i^{(t)} \rangle) |
\end{align}
where
\begin{align*}
&|\mu_i( \langle \mathbf{p},
\boldsymbol{\theta}_{i,0}\rangle) - \mu_i(\langle \mathbf{p},
\widetilde{\boldsymbol{\theta}}_i^{(t)} \rangle) | \\
&\leq  L_{\mu_i}|\langle \mathbf{p}, \boldsymbol{\theta}_{i,0}-
\widetilde{\boldsymbol{\theta}}_i^{(t)} \rangle| \quad \text{(by Assumption \ref{assumption_mu})}\\
&= L_{\mu_i} |\mathbf{p}^{\top}  G_i^{(t-1),-1}(\boldsymbol{\theta}_{i,0},\widetilde{\boldsymbol{\theta}}_i^{(t)})(g_i^{(t-1)}(\boldsymbol{\theta}_{i,0})-g_i^{(t-1)}(\widetilde{\boldsymbol{\theta}}_i^{(t)}))| \quad \text{(by \eqref{eq_g_t_G_t})} \\
&\leq L_{\mu_i} \lVert \mathbf{p} \rVert_{G_i^{(t-1),-1}(\boldsymbol{\theta}_{i,0},\widetilde{\boldsymbol{\theta}}_i^{(t)})} \lVert g_i^{(t-1)}(\boldsymbol{\theta}_{i,0})-g_i^{(t-1)}(\widetilde{\boldsymbol{\theta}}_i^{(t)})
\rVert_{G_i^{(t-1),-1}(\boldsymbol{\theta}_{i,0},\widetilde{\boldsymbol{\theta}}_i^{(t)})} \quad \text{(by Cauchy-Schwartz inequality)}\\
&\leq \frac{L_{\mu_i}}{c_{\mu_i}}  \lVert \mathbf{p} \rVert_{V_{i}^{(t-1),-1}} \lVert g_i^{(t-1)}(\boldsymbol{\theta}_{i,0})-g_i^{(t-1)}(\widetilde{\boldsymbol{\theta}}_i^{(t)})\rVert_{V_{i}^{(t-1),-1}} \quad \text{(by \eqref{eq_G_t})} \\
&  \leq \frac{2 L_{\mu_i}}{c_{\mu_i}}  \lVert \mathbf{p} \rVert_{V_{i}^{(t-1),-1}} \lVert g_i^{(t-1)}(\boldsymbol{\theta}_{i,0})-g_i^{(t-1)}(\widehat{\boldsymbol{\theta}}_i^{(t)}
)\rVert_{V_{i}^{(t-1),-1}} \quad \text{(by definition of $\widetilde{\boldsymbol{\theta}}_i^{(t)}$)}.
\end{align*}

Now we calculate

\begin{align*}
\lVert g_i^{(t-1)}(\boldsymbol{\theta}_{i,0})-g_i^{(t-1)}(\widehat{\boldsymbol{\theta}}_i^{(t)}
)\rVert_{V_{i}^{(t-1),-1}}
&\leq
\left\| \sum_{s= 1}^{t-1}  \mu_i(\mathbf{p}^{(s)\top} \boldsymbol{\theta}_{i,0}) \mathbf{p}^{(s)}
+ \lambda \boldsymbol{\theta}_{i,0} - \sum_{s= 1}^{t-1} \mathbf{p}^{(s)} y_i^{(s)}  
\right\|_{V_{i}^{(t-1),-1}}  \\
&\leq \left\|
-\sum_{s= 1}^{t-1} \mathbf{p}^{(s)} \eta_i^{(s)}  + \lambda \boldsymbol{\theta}_{i,0}
\right\|_{V_{i}^{(t-1),-1}}  \\
&= 
\left\| -S_i^{(t-1)}  + \lambda \boldsymbol{\theta}_{i,0} \right\|_{V_{i}^{(t-1),-1}} \\
&\leq  \lVert  S_i^{(t-1)} \rVert_{V_{i}^{(t-1),-1}}  + \lVert \lambda \boldsymbol{\theta}_{i,0} \rVert_{V_{i}^{(t-1),-1}} \quad \text{(Triangle inequality)}\\
&\leq \lVert  S_i^{(t-1)} \rVert_{V_{i}^{(t-1),-1}}  + \sqrt{\lambda c_{\mu_i}} \lVert \boldsymbol{\theta}_{i,0} \rVert_{2}  \quad \left(V_{i}^{(t-1)} \geq \frac{\lambda}{c_{\mu_i}} I_N\right) \\
&\leq \underset{=\mathrm{c}_i^{(t-1)}(\delta)}{\underbrace{L_{\mu_i} \sqrt{2\log(T/\delta) + N \log\left( 1 + \frac{B_p^2 (t-1)}{N\lambda}\right)}}}+ B_{\theta_i}\sqrt{\lambda c_{\mu_i}} \quad (\text{w.p.}\geq 1-\delta)  .
\end{align*}
   
In the last inequality, we have used the assumption that $\lVert \boldsymbol{\theta}_{i,0} \rVert_2 \leq B_{\theta_i}$ and the concentration result established in Proposition 7 of \citet{russac2019weighted} for the self-normalized quantity $\| S_i^{(t-1)} \|_{V_{i}^{(t-1),-1}}$ (which can be applied thanks to the conditional $L_{\mu_i}$-subgaussianity of $\eta_i^{(t)}$ established in \cref{lemma:subgaussian_noise}\footnote{In \citet{russac2019weighted}, the subgaussianity of the error term is the sole assumption required to prove their Proposition 5, which in turn implies Proposition 7.}), which establishes that

$$
\mathbb{P}\left(\exists t \leq T,\|S_i^{(t)}\|_{V^{(t),{-1}}} \geq \mathrm{c}^{(t)}_i(\delta)\right) \leq \delta, \quad \text{equivalently}\quad \mathbb{P}\left(\|S_i^{(t)}\|_{V^{(t),{-1}}} \leq \mathrm{c}^{(t)}_i(\delta), \forall t \leq T\right) \geq 1-\delta.
$$

We established that, for any $\mathbf{p} \in \mathcal{P}$, 
$$
\mathbb{P}\left(|\operatorname{rev}_i(p_i | \mathbf{p}_{-i}) - \widehat{ \operatorname{rev}}_i^{(t)}(p_i  | \mathbf{p}_{-i})| \leq p_i \lVert \mathbf{p} \rVert_{V_{i}^{(t-1),-1}} \cdot \underset{=\rho_i^{(t-1)}(\delta)}{\underbrace{\frac{2 L_{\mu_i}}{c_{\mu_i}}\left(\mathrm{c}_i^{(t-1)}(\delta)+B_{\theta_i}\sqrt{\lambda c_{\mu_i}}\right)}},\quad \forall t \leq T \right) \geq 1-\delta.
$$

This completes the proof.

\section{Proof of Theorem \ref{theorem_regret_SW}}
We recall the definitions
$$
\operatorname{rev}_i(p_i | \mathbf{p}_{-i}) = p_i\mu_i( -\beta_{i,0}p_i+\langle \boldsymbol{\gamma}_{i,0},\mathbf{p}_{-i}\rangle), \qquad \forall (p_i , \mathbf{p}_{-i}) \in \mathcal{P},
$$
and
$$
\widehat{ \operatorname{rev}}_i^{(t)}(p_i | \mathbf{p}_{-i}) =  p_i\mu_{i}( -\widetilde{\beta}_i^{(t)\nML}p_i+\langle\widetilde{\boldsymbol{\gamma}}_i^{(t)\nML}, \mathbf{p}_{-i}\rangle), \qquad \forall (p_i , \mathbf{p}_{-i}) \in \mathcal{P}.
$$
For simplicity of notation, we also define
$$
\textstyle p_i^{(t)\star} =\Gamma_i(\mathbf{p}^{(t)}_{-i})=  \underset{p_i \in \mathcal{P}_i}{\operatorname{argmax}} \, \operatorname{rev}_i(p_i | \mathbf{p}_{-i}^{(t)}) , \qquad \forall  \mathbf{p}_{-i} ^{(t)}\in \mathcal{P}_{-i},
$$
and 
$$
\textstyle p_i^{(t)} =\widehat{\Gamma}^{(t)}_i(\mathbf{p}^{(t-1)}_{-i}) \in  \underset{p_i \in \mathcal{P}_i}{\operatorname{argmax}} \, \text{UCB}_i^{(t)}(\cdot |\mathbf{p}_{-i}^{(t-1)}) ,\qquad \forall  \mathbf{p}_{-i} ^{(t-1)}\in \mathcal{P}_{-i},
$$
where we recall from \cref{def:UCB_conditioned} that
\begin{align*}
\mathrm{UCB}_i^{(t)}(p_i| \mathbf{p}_{-i}^{(t-1)})= \widehat{ \operatorname{rev}}_i^{(t)}(p_i | \mathbf{p}_{-i}^{(t-1)}) +\sigma_i^{(t-1)}(p_i,\delta).
\end{align*}
\noindent
The goal is to upper bound the individual regret
$$
\operatorname{Reg}_i(T)=\mathbb{E} \sum_{t=1}^T R_i^{(t)} =\mathbb{E} \sum_{t=1}^T[  \operatorname{rev}_i(\Gamma_i(\mathbf{p}^{(t)}_{-i}) | \mathbf{p}_{-i}^{(t)}) -  \operatorname{rev}_i(p^{(t)}_i | \mathbf{p}_{-i}^{(t)}) ].
$$ 
We have
\begin{align}\label{eq:rev_diff}
R_i^{(t)} = \operatorname{rev}_i(\Gamma_i(\mathbf{p}^{(t)}_{-i}) | \mathbf{p}_{-i}^{(t)}) -  \operatorname{rev}_i(p^{(t)}_i | \mathbf{p}_{-i}^{(t)}) \leq C_{i,1}(\Gamma_i(\mathbf{p}^{(t)})-p_i^{(t)})^2,
\end{align}
where we used that $R_i^{(t)}\leq M_i$ and that
\begin{align*}
|\partial^2_{p^2}\operatorname{rev}_i(p | \mathbf{p}_{-i}^{(t)})_{|p=p'}| &= |-2\beta_{i,0}\mu_i'(-\beta_{i,0}p'+\boldsymbol{\gamma}_{i,0}^{\top}\mathbf{p}_{-i}^{(t)}) +\beta_{i,0}^2p'\mu_i''(-\beta_{i,0}p'+\boldsymbol{\gamma}_{i,0}^{\top}\mathbf{p}_{-i}^{(t)})|\\
&\leq 2\beta_{i,0}L_{\mu_i}+\beta_{i,0}^2\bar{p_i}B''_i \\
&\leq 2\overline{\beta}_iL_{\mu_i}+\overline{\beta}_i^2\bar{p_i}B''_i\triangleq C_{i,1},
\end{align*}
for some $p'$ in the segment between the points $p_i^{(t)}$ and $\Gamma_i(\mathbf{p}^{(t)}_{-i})$. Continuing from \eqref{eq:rev_diff} we have
\begin{align}\label{eq:rev_diff_2}
(\Gamma_i(\mathbf{p}^{(t)}_{-i})-p_i^{(t)})^2
&\leq \|\boldsymbol{\Gamma}(\mathbf{p}^{(t)})-\mathbf{p}^{(t)}\|_2^2\nonumber\\
&\leq 2\|\boldsymbol{\Gamma}(\mathbf{p}^{(t)})-\mathbf{p}^{\star}\|_2^2 +2\|\mathbf{p}^{\star}-\mathbf{p}^{(t)}\|_2^2\nonumber\\
&= 2(\|\boldsymbol{\Gamma}(\mathbf{p}^{(t)})-\boldsymbol{\Gamma}(\mathbf{p}^{\star})\|_2^2 +\|\mathbf{p}^{\star}-\mathbf{p}^{(t)}\|_2^2), \qquad &[\boldsymbol{\Gamma}(\mathbf{p}^{\star})=\mathbf{p}^{\star} \text{ by }\cref{strong-concavity-rev}]\nonumber\\
&\leq 2(L_{\boldsymbol{\Gamma}}\|\mathbf{p}^{(t)}-\mathbf{p}^{\star}\|_2^2 +\|\mathbf{p}^{\star}-\mathbf{p}^{(t)}\|_2^2) \quad &[\text{by }\cref{strong-concavity-rev}]\nonumber\\
&\leq 2(L_{\boldsymbol{\Gamma}}+1)\|\mathbf{p}^{(t)}-\mathbf{p}^{\star}\|_2^2.
\end{align}

\paragraph{STEP 1: Bound for $\|\mathbf{p}^{(t)}-\mathbf{p}^{\star}\|_2^2$.}
Consider the event
\begin{equation}\label{set:G}
\mathcal{G}=\cap_{t \in [T]}\left\{\|\widehat{\boldsymbol{\Gamma}}^{(t)}(\mathbf{p}^{(t-1)})-\boldsymbol{\Gamma}(\mathbf{p}^{(t-1)})\|_2^2 \leq C_2 \sigma^{(t-1)}(\delta)\right\},
\end{equation}
that has probability at least $1-2N\delta$ by \cref{prop_SW_GLM_anytime_upper}. Continuing from \eqref{eq:rev_diff_2}, we get that, on $\mathcal{G}$,

\begin{align}\label{ineq:R_i_upper}
\|\mathbf{p}^{(t)}-\mathbf{p}^{\star}\|_2 &= \|\widehat{\boldsymbol{\Gamma}}^{(t)}(\mathbf{p}^{(t-1)})-\mathbf{p}^{\star}\|_2 \nonumber\\
&\overset{(\star)}{\leq} \|\widehat{\boldsymbol{\Gamma}}^{(t)}(\mathbf{p}^{(t-1)})-\boldsymbol{\Gamma}(\mathbf{p}^{(t-1)})\|_2+\|\boldsymbol{\Gamma}(\mathbf{p}^{(t-1)})-\boldsymbol{\Gamma}(\mathbf{p}^{\star})\|_2\nonumber \\
&\leq \sqrt{C_2\sigma^{(t-1)}(\delta)} +L_{\boldsymbol{\Gamma}}\|\mathbf{p}^{(t-1)}-\mathbf{p}^{\star}\|_2\nonumber \\
&\leq \sqrt{C_2\sigma^{(t-1)}(\delta)}+L_{\boldsymbol{\Gamma}}\left(\sqrt{C_2\sigma^{(t-2)}(\delta)}+L_{\boldsymbol{\Gamma}}\|\mathbf{p}^{(t-2)}-\mathbf{p}^{\star}\|_2\right)\nonumber \\
&= \sqrt{C_2}\left(\sqrt{\sigma^{(t-1)}(\delta)} +L_{\boldsymbol{\Gamma}}\sqrt{\sigma^{(t-2)}(\delta)}\right)+L^2_{\boldsymbol{\Gamma}}\|\mathbf{p}^{(t-2)}-\mathbf{p}^{\star}\|_2\nonumber \\
&=  \sqrt{C_2}\left(\sum_{j=0}^{\ell}L^{j}_{\boldsymbol{\Gamma}} \sqrt{\sigma^{(t-j-1)}(\delta)}\right)+L^{\ell}_{\boldsymbol{\Gamma}}\|\mathbf{p}^{(t-\ell)}-\mathbf{p}^{\star}\|_2, \qquad \ell \in \{1,\dots, t-1\} \nonumber \\
&\leq \sqrt{C_2}\left(\sum_{j=0}^{t-1}L^{j}_{\boldsymbol{\Gamma}} \sqrt{\sigma^{(t-j-1)}(\delta)}\right)+L^{t-1}_{\boldsymbol{\Gamma}}\|\mathbf{p}^{(0)}-\mathbf{p}^{\star}\|_2.
\end{align}
Summing over $t$ we obtain, on $\mathcal{G}$, the following inequality

\begin{align}\label{ineq:NE_convergence}
\sum_{t=1}^T\|\mathbf{p}^{(t)}-\mathbf{p}^{\star}\|_2^2 & \leq \sum_{t=1}^T \left\{2C_2\left(\sum_{j=0}^{t-1} L^{j}_{\boldsymbol{\Gamma}}\sqrt{\sigma^{(t-j-1)}(\delta)}\right)^2+2L^{2(t-1)}_{\boldsymbol{\Gamma}}\|\mathbf{p}^{(0)}-\mathbf{p}^{\star}\|_2^2 \right\}\nonumber\\
&\overset{(\star\star)}{\leq} \frac{2 C_2}{(1-L_{\boldsymbol{\Gamma}})^2}\sum_{t=0}^{T-1} \sigma^{(t)} (\delta)+2\left(\sum_{t=1}^TL^{2(t-1)}_{\boldsymbol{\Gamma}}\right)\|\mathbf{p}^{(0)}-\mathbf{p}^{\star}\|_2^2 \nonumber\\
&\leq \frac{2 C_2}{(1-L_{\boldsymbol{\Gamma}})^2}\sum_{t=0}^{T-1} \sigma^{(t)}(\delta)  + \frac{2\|\mathbf{p}^{(0)}-\mathbf{p}^{\star}\|_2^2}{1-L^{2}_{\boldsymbol{\Gamma}}}\nonumber \\
&= C_4 \sum_{t=0}^{T-1} \sum_{j \in [N]} 2\rho_j^{(t)}(\delta) p_j^{(t+1)}\| (p_j^{(t+1)},\mathbf{p}_{-j}^{(t)} )\|_{V_{j}^{(t),-1}} + C_3,
\end{align}

where $C_3 = 2 \frac{\|\mathbf{p}^{(0)}-\mathbf{p}^{\star}\|_2^2}{1-L^{2}_{\boldsymbol{\Gamma}}}$ and $C_4=\frac{2 C_2}{(1-L_{\boldsymbol{\Gamma}})^2}$. In $(\star\star)$ we used the following \cref{lemma:tecnincal_summation}, by substituting $q\leftarrow L_{\Gamma}$ and $a_k\leftarrow \sqrt{\sigma^{(k)}(\delta)}$ (the proof is deferred to \cref{sec:lemma:tecnincal_summation}).

\begin{lemma}\label{lemma:tecnincal_summation}
Let $(a_k)_{k\ge 0}$ be any real sequence, $0<q<1$, and $T\ge 1$. Then
\[
\sum_{t=1}^T\Biggl(\sum_{j=0}^{t-1} q^j a_{t-j-1}\Biggr)^2
 \le  \frac{1}{(1-q)^2} \sum_{k=0}^{T-1} a_k^2.
\]
\end{lemma} 

\paragraph{STEP 2: Regret Bound.}
Consider the set $\mathcal{G}$ defined in \eqref{set:G} and the partitioning of the sampling space: $\mathcal{G} \cup \mathcal{G}^c$. We can write
$$
\begin{aligned}
\mathbb{E}\left[R_i^{(t)}\right] & =\mathbb{E}\left[R_i^{(t)} \cdot\left(\mathbb{I}(\mathcal{G})+\mathbb{I}\left(\mathcal{G}^c\right)\right)\right]
\\
&\leq\mathbb{E}\left[R_i^{(t)} \cdot\mathbb{I}(\mathcal{G})\right] + M_i \mathbb{P}\left(\mathcal{G}^c\right)
\\
&\leq\mathbb{E}\left[R_i^{(t)} \cdot\mathbb{I}(\mathcal{G})\right] + M_i 2N \delta.
\end{aligned}
$$

Where $M_i$ is an upper bound of $R_i^{(t)}$. For example, using that $0\leq \mu_i \leq B_i$ and $0 \leq p_i\leq \bar{p}_i$, we can take $M_i = 2 B_i \bar{p}_i$. Then 
$$
\begin{aligned}
\mathbb{E}\left[\sum_{t=1}^T R_i^{(t)}\right] \leq \mathbb{E}\left[\sum_{t=1}^T R_i^{(t)} \cdot\mathbb{I}(\mathcal{G})\right] + M_i 2N \delta T.
\end{aligned}
$$

Now, choosing $\delta = 1/T^{\gamma}$ with $\gamma> 1$, the second term converges to zero, so we only need to bound the first term on $\mathcal{G}$. Merging \eqref{eq:rev_diff} and \eqref{eq:rev_diff_2} we have that, on $\mathcal{G}$, up to a multiplicative constant $C_4$ and an additive constant $C_3$

\begin{align}\label{ineq:regret-nash}
\operatorname{Reg}_i(T) = \mathbb{E}\left[\sum_{t=1}^T R_i^{(t)} \cdot\mathbb{I}(\mathcal{G})\right] & \lesssim \mathbb{E}\left[\sum_{t=1}^T\|\mathbf{p}^{(t)}-\mathbf{p}^{\star}\|_2^2 \cdot\mathbb{I}(\mathcal{G})\right]\nonumber\\
& \overset{\eqref{ineq:NE_convergence}}{\lesssim} \mathbb{E}\left[\sum_{t=0}^{T-1} \sum_{j \in [N]} \rho_j^{(t)}(\delta)p_j^{(t+1)} \| (p_j^{(t+1)},\mathbf{p}_{-j}^{(t)} )\|_{V_{j}^{(t),-1}} \cdot\mathbb{I}(\mathcal{G})\right]\nonumber \\
& \leq \bar{p} \,\mathbb{E}\left[\sum_{t=0}^{T-1} \sum_{j \in [N]} \rho_j^{(t)}(\delta) \| (p_j^{(t+1)},\mathbf{p}_{-j}^{(t)} )\|_{V_{j}^{(t),-1}}\cdot\mathbb{I}(\mathcal{G})\right],
\end{align}

where $\bar{p}= \max_{j \in [N]}\bar{p}_j >0$ and where we used that $\rho_j^{(t)}$ (defined in \eqref{def:UCB}) is increasing in $t$. Moreover, for $\delta = 1/T^{\gamma}$ with $\gamma > 1$, (retaining the dependence on $\lambda,N$ and $T$), we have 
\begin{align*}
\rho_j^{(t)}(\delta) = O\left( \mathrm{c}^{(t)}_i(\delta)+\sqrt{\lambda}\right)&=O\left(\sqrt{\lambda}+L_{\mu_i} \sqrt{2 \log \left(\frac{T}{\delta}\right)+N \log \left(1+\frac{B_p^2 T}{N \lambda}\right)}\right)\\
& = O\left(\sqrt{\lambda}+\sqrt{\log(T)+N \log \left(1+\frac{T}{N\lambda}\right)}\right).
\end{align*}

It only remains to bound 
$$
\sum_{t=0}^{T-1} \sum_{j \in [N]} \| (p_j^{(t+1)},\mathbf{p}_{-j}^{(t)} )\|_{V_{j}^{(t),-1}},
$$
for which we use the following \cref{lemma:Elliptical-Lemma}, proved in \cref{subsec:lemma:Elliptical-Lemma}.

\begin{lemma}[Variant of the Elliptical Potential Lemma]\label{lemma:Elliptical-Lemma} Retaining only the dependence on $\lambda, N$ and $T$, it holds for $\lambda_j \geq c_{\mu_j} B_p^2$
$$
\sum_{t=0}^{T-1} \sum_{j \in [N]} \| (p_j^{(t+1)},\mathbf{p}_{-j}^{(t)} )\|_{V_{j}^{(t),-1}} = O \left(N^{3/2}\sqrt{T \log \left(\frac{T}{N \lambda}+1\right)}\right).
$$
\end{lemma}

We then conclude that the regret, in terms of $\lambda, N$ and $T$ is

\begin{align}\label{eq:regret_dependence_on_lambda}
\operatorname{Reg}_i(T) = O\left(\left[\sqrt{\lambda}+\sqrt{\log(T)+N \log \left(1+\frac{T}{N\lambda}\right)}\right]\cdot N^{3/2}\sqrt{T \log \left(\frac{T}{N \lambda}+1\right)}\right).
\end{align}

For any $1\leq \lambda \leq O(N\log(T))$ we have

$$
\begin{aligned}
\operatorname{Reg}_i(T) = O\left(\sqrt{N\log(T)}\cdot N^{3/2}\sqrt{T\log(T)}\right)=O\left(N^{2}\sqrt{T}\log(T)\right).
\end{aligned}
$$

Note that from \eqref{ineq:regret-nash} we can also recover convergence to NE.

$$
\sum_{t=1}^T\mathbb{E}[\|\mathbf{p}^{(t)}-\mathbf{p}^{\star}\|_2^2] = O(N^{2}\sqrt{T}\log(T)).
$$

This completes the proof.

\section{Proof of \cref{lemma:tecnincal_summation}}\label{sec:lemma:tecnincal_summation}
We want to prove that for a given sequence $(a_k)_{k\ge 0}$ of real values, $0<q<1$, and $T\ge 1$, then
\[
\sum_{t=1}^T\Biggl(\sum_{j=0}^{t-1} q^j a_{t-j-1}\Biggr)^2
 \le  \frac{1}{(1-q)^2} \sum_{k=0}^{T-1} a_k^2.
\]
To show this, we expand the square and re-index with $k=t-j-1$ and $m=t-\ell-1$, which yields
\[
\sum_{t=1}^T\Biggl(\sum_{j=0}^{t-1} q^j a_{t-j-1}\Biggr)^2
= \sum_{k=0}^{T-1}\sum_{m=0}^{T-1} a_k a_m \sum_{t=\max\{k,m\}+1}^{T} q^{2t-k-m-2}.
\]
For $0<q<1$ we have
\[
\sum_{t=\max\{k,m\}+1}^{T} q^{2t-k-m-2}
= q^{ 2\max\{k,m\}-k-m}\cdot \sum_{r=0}^{T-\max\{k,m\}-1} q^{2r}
 \leq  \frac{q^{ 2\max\{k,m\}-k-m}}{1-q^2}.
\]
Hence
\[
\sum_{t=1}^T\Biggl(\sum_{j=0}^{t-1} q^j a_{t-j-1}\Biggr)^2
 \le  \frac{1}{1-q^2}\sum_{k=0}^{T-1}\sum_{m=0}^{T-1} q^{ 2\max\{k,m\}-k-m} a_k a_m.
\]
Observe that $q^{2\max\{k,m\}-k-m}=q^{|k-m|}$, so the right-hand side equals
\[
\frac{1}{1-q^2}  \mathbf{a}^\top K \mathbf{a}, \qquad
K_{k,m}=q^{|k-m|}, \qquad \mathbf{a}=(a_0,\dots,a_{T-1})^\top .
\]
By the Schur test,
\[
\|K\|_2  \le  \max_{k}\sum_{m=0}^{T-1} q^{|k-m|}
 \le  1 + 2\sum_{r=1}^{\infty} q^{r}
 =  \frac{1+q}{1-q}.
\]
Therefore,
\[
\sum_{t=1}^T\Biggl(\sum_{j=0}^{t-1} q^j a_{t-j-1}\Biggr)^2
 \le  \frac{1}{1-q^2} \|K\|_2 \|\mathbf{a}\|_2^2
 \le  \frac{1}{1-q^2}\cdot \frac{1+q}{1-q} \sum_{k=0}^{T-1} a_k^2
 =  \frac{1}{(1-q)^2} \sum_{k=0}^{T-1} a_k^2,
\]
since $1-q^2=(1-q)(1+q)$. This proves the claim.

\section{Proof of \cref{lemma:Elliptical-Lemma}}
\label{subsec:lemma:Elliptical-Lemma}
Recall from \eqref{eq_Design_matrix_SW} that
$$
V_{j}^{(t)} = \sum_{s=1}^{t} \mathbf{p}^{(s)} \mathbf{p}^{(s)\top} + \frac{\lambda}{c_{\mu_j}} I_N,
$$
First, note:
$$
(p_j^{(t+1)}, \mathbf{p}_{-j}^{(t)})=(p_j^{(t)}, \mathbf{p}_{-j}^{(t)})+(\underbrace{p_j^{(t+1)}-p_j^{(t)}}_{\Delta p_j^{(t)}}, \boldsymbol{0}_{-j}) = \mathbf{p}^{(t)}+ \Delta p_j^{(t)} \mathbf{e}_j,
$$
where $\mathbf{e}_j$ is the null vector with only entry $1$ in the $j$-th coordinate. So:
$$
\|(p_j^{(t+1)}, \mathbf{p}_{-j}^{(t)})\|_{V_j^{(t),-1}} \leq \|\mathbf{p}^{(t)}\|_{V_j^{(t),-1}}+\|\Delta p_j^{(t)} \mathbf{e}_j\|_{V_j^{(t),-1}}\leq \|\mathbf{p}^{(t)}\|_{V_j^{(t),-1}}+\sqrt{\Delta p_j^{(t)}[V_j^{(t),-1}]_{jj}\Delta p_j^{(t)}}.
$$
Note that $\sqrt{\Delta p_j^{(t)}[V_j^{(t),-1}]_{jj}\Delta p_j^{(t)}} = \|\Delta p_j^{(t)}\|_{[V_j^{(t),-1}]_{jj}}$ is a norm, hence, by triangula inequality 
$$
\sqrt{\Delta p_j^{(t)}[V_j^{(t),-1}]_{jj}\Delta p_j^{(t)}} = \|\Delta p_j^{(t)}\|_{[V_j^{(t),-1}]_{jj}}\leq \|p_j^{(t+1)}\|_{[V_j^{(t),-1}]_{jj}} + \|p_j^{(t)}\|_{[V_j^{(t),-1}]_{jj}} \leq \|\mathbf{p}^{(t+1)}\|_{V_j^{(t),-1}} + \|\mathbf{p}^{(t)}\|_{V_j^{(t),-1}}
$$
Hence
$$
\sum_{t=0}^{T-1} \sum_{j \in [N]}\|(p_j^{(t+1)}, \mathbf{p}_{-j}^{(t)})\|_{V_j^{(t),-1}}\leq 2\underset{=A}{\underbrace{\sum_{j \in [N]} \sum_{t=0}^{T-1}  \|\mathbf{p}^{(t)}\|_{V_j^{(t),-1}}}}+\underset{=B}{\underbrace{\sum_{j \in [N]} \sum_{t=0}^{T-1}\|\mathbf{p}^{(t+1)}\|_{V_j^{(t),-1}}}}.
$$
Note that

\paragraph{Bound B.}

Since $\|\mathbf{p}^{(t)}\|_2\leq B_p$, for $\lambda_j \geq c_{\mu_j} B_p^2$ we have $\mathbf{p}^{(t)}\mathbf{p}^{(t) \top} \preccurlyeq B_p^2 I \preccurlyeq  \frac{\lambda}{c_{\mu_j}} I$. Hence
$$
V_j^{(t+1)}= V_j^{(t)}+\mathbf{p}^{(t)}\mathbf{p}^{(t) \top} \preccurlyeq V_j^{(t)}+\frac{\lambda}{c_{\mu_j}} I \preccurlyeq V_j^{(t)} + \mathbf{p}^{(t-1)}\mathbf{p}^{(t-1) \top}+\frac{\lambda}{c_{\mu_j}} I \preccurlyeq 2 V_j^{(t)},
$$
which implies $V_j^{(t),-1} \preccurlyeq 2 V_j^{(t+1),-1}$, hence
$$
\sum_{j \in [N]} \sum_{t=0}^{T-1}  \|\mathbf{p}^{(t+1)}\|_{V_j^{(t),-1}} \leq \sqrt{2}\sum_{j \in [N]} \sum_{t=0}^{T-1}  \|\mathbf{p}^{(t+1)}\|_{V_j^{(t+1),-1}}
$$
i.e. $B \leq \sqrt{2} A$ (up to rescaling the sum $\sum_{t=0}^{T-1}$ to $\sum_{t=1}^T$, which only affects the final sum $\sum_{t=0}^{T-1} \sum_{j \in [N]}\|(p_j^{(t+1)}, \mathbf{p}_{-j}^{(t)})\|_{V_j^{(t),-1}}$ by a constant).

\paragraph{Bound A.}
We use the Elliptical Potential Lemma in \citet{carpentier2020elliptical} Proposition 1 (by replacing $p\leftarrow 1$ and $\lambda \leftarrow \frac{\lambda}{c_{\mu_j}}$):

\begin{align*}
A=\sum_{j \in [N]}\sum_{t=0}^{T-1} \| \mathbf{p}^{(t)}\|_{V_j^{(t),-1}} & \leq \sum_{j \in [N]} \sqrt{T N \log \left(\frac{c_{\mu_j}T}{N \lambda}+1\right)}\leq N \sqrt{T N \log \left(T\frac{c_{\mu}}{N \lambda}+1\right)}.
\end{align*}

\paragraph{Merging the two bounds.}
Retaining only the dependence on $\lambda, N$ and $T$, we recover 
$$
\sum_{t=0}^{T-1} \sum_{j \in [N]}\|(p_j^{(t+1)}, \mathbf{p}_{-j}^{(t)})\|_{V_j^{(t),-1}} = O\left(N\sqrt{T N \log \left(\frac{T}{N \lambda}+1\right)}\right).
$$

This completes the proof.

\section{Proof of \cref{lemma:subgaussian_noise}}
\label{subsec:hoeffding_cond}
We want to prove that, under \cref{assumption_mu}, for every $i \in [N]$, $\eta_i^{(t)}$ is $L_{\mu_i}$-subgaussian conditionally on $ \mathcal{H}_i^{(t-1)}$. Fix $i \in [N]$. We first prove that $\mathbb{E} [\eta_i^{(t)} | \mathcal{H}^{(t-1)}_i] =0$, where 
$$
\eta_i^{(t)} = y_i^{(t)} - \mu_i(\langle\mathbf{p}^{(t)}, \boldsymbol{\theta}_{i,0}\rangle).
$$
We recall that $p_i^{(t)}$ is computed using information $\mathcal{H}_i^{(t-1)} = \{(y_i^{(s)},\mathbf{p}^{(s)})\}_{s \leq t-1}$. Let 
$$
\mathcal{F}_i^{(t-1)}=\sigma\left( \mathcal{H}^{(t-1)}_i , \mathbf{p}^{(t)}\right).
$$
Since $\mathcal{H}_i^{(t-1)} \subseteq \mathcal{F}_i^{(t-1)}$, we have by the law of total expectation
\begin{align*}
\mathbb{E} [\eta_i^{(t)} | \mathcal{H}^{(t-1)}_i] &= \mathbb{E} [\mathbb{E} [\eta_i^{(t)} | \mathcal{F}^{(t-1)}_i]| \mathcal{H}^{(t-1)}_i] \\
&= \mathbb{E} [\mathbb{E} [y_i^{(t)} - \mu_i(\langle\mathbf{p}^{(t)}, \boldsymbol{\theta}_{i,0}\rangle) | \mathcal{F}^{(t-1)}_i]| \mathcal{H}^{(t-1)}_i]\\
&= \mathbb{E} [\mu_i(\langle\mathbf{p}^{(t)}, \boldsymbol{\theta}_{i,0}\rangle) - \mu_i(\langle\mathbf{p}^{(t)}, \boldsymbol{\theta}_{i,0}\rangle)| \mathcal{H}^{(t-1)}_i]\\
&=0.
\end{align*}
In the third equality we used that, by construction $y_i^{(t)}$ is independent of $\mathcal{H}_{i}^{(t-1)}$ given $\mathbf{p}^{(t)}$, i.e. $\mathbb{E}[y_i^{(t)}\mid \mathcal{H}_{i}^{(t-1)}, \mathbf{p}^{(t)}]=\mathbb{E}[y_i^{(t)}\mid\mathbf{p}^{(t)}] = \mu_i(\langle\mathbf{p}^{(t)}, \boldsymbol{\theta}_{i,0}\rangle)$. Now let $w_i^{(t)} \equiv\left\langle\boldsymbol{\theta}_{i, 0}, \mathbf{p}^{(t)}\right\rangle$ (this is $\mathcal{F}_i^{(t-1)}$-measurable), and recall from \eqref{eq:expon_model} that we model

\begin{equation*}
\frac{d\mathbb{P}_{\boldsymbol{\theta}_{i,0}}(y_i^{(t)} | \mathbf{p}^{(t)})}{d\nu_i(y_i^{(t)})} = \exp\left\{
y_i^{(t)} w_i^{(t)} - b_i(w_i^{(t)}) + c_i(y_i^{(t)})\right\},
\end{equation*}

Then, with $\mu_i(w_i^{(t)})=b_i^{\prime}(w_i^{(t)})$ and $\eta_i^{(t)}=y_i^{(t)}-\mu_i(w_i^{(t)})$,

$$
\begin{aligned}
\mathbb{E}\left[e^{\lambda \eta_i^{(t)}} | \mathcal{F}_i^{(t-1)}\right] & =e^{-\lambda \mu_i( w_i^{(t)})} \int e^{\lambda  y_i^{(t)}} \exp \left\{ y_i^{(t)}  w_i^{(t)}-b_i( w_i^{(t)})+c_i( y_i^{(t)})\right\} d \nu_i( y_i^{(t)}) \\
& =e^{-\lambda \mu_i( w_i^{(t)})} e^{-b_i( w_i^{(t)})} \int \exp \left\{ y_i^{(t)}( w_i^{(t)}+\lambda)+c_i( y_i^{(t)})\right\} d \nu_i( y_i^{(t)}) \\
& =e^{-\lambda \mu_i( w_i^{(t)})} e^{-b_i( w_i^{(t)})} e^{b_i( w_i^{(t)}+\lambda)}=\exp \left\{b_i( w_i^{(t)}+\lambda)-b_i( w_i^{(t)})-b_i^{\prime}( w_i^{(t)}) \lambda\right\}.
\end{aligned}
$$

Taking logs gives the identity
\[
\log \mathbb E \left[\exp \big(\lambda \eta_i^{(t)}\big) \middle|  \mathcal{F}_i^{(t-1)}\right]
= b_i( w_i^{(t)}+\lambda)-b_i( w_i^{(t)})-b_i^{\prime}( w_i^{(t)}) \lambda.
\]
By Taylor’s expansion, there exists a $\xi \in [0,1]$ such that
\[
b_i( w_i^{(t)}+\lambda)-b_i( w_i^{(t)})-b_i^{\prime}( w_i^{(t)}) \lambda
= \tfrac{1}{2} b_i''( w_i^{(t)}+\xi\lambda) \lambda^2
= \tfrac{1}{2} \mu_i'( w_i^{(t)}+\xi\lambda) \lambda^2.
\] Assumption~\ref{assumption_mu} yields $\mu_i'(\cdot)\le L_{\mu_i}$ on $\mathbb{R}$, hence
\[
\log \mathbb E \left[\exp \big(\lambda \eta_i^{(t)}\big) \middle|  \mathcal{F}_i^{(t-1)}\right]
\leq \frac{L_{\mu_i} \lambda^2}{2} \quad \implies \quad \mathbb E \left[\exp \big(\lambda \eta_i^{(t)}\big) \middle|  \mathcal{F}_i^{(t-1)}\right]\leq \exp \left( \frac{L_{\mu_i} \lambda^2}{2} \right).
\]
And consequently 
\begin{equation*}
\mathbb E \left[\exp \big(\lambda \eta_i^{(t)}\big) \middle| \mathcal{H}_i^{(t-1)}\right]=\mathbb E \left[\mathbb E \left[\exp \big(\lambda \eta_i^{(t)}\big) \middle| \mathcal{F}_i^{(t-1)}\right]\middle| \mathcal{H}_i^{(t-1)}\right]
\leq \exp \left(\tfrac{L_{\mu_i}\lambda^2}{2}\right)
\quad\text{a.s. for all $\lambda \in \mathbb{R}$}.
\end{equation*}
This proves that \(\eta_i^{(t)}\) is conditionally sub-gaussian with variance proxy \(L_{\mu_i}\). This completes the proof.

\end{document}